\documentclass[reqno,a4paper]{amsart}

\usepackage{amsmath,amsthm,amssymb,amsfonts,enumerate,graphicx,color,pstricks,marginnote}
\numberwithin{equation}{section} 
\theoremstyle{plain}
\newtheorem{theo+}           {Theorem}      [section]
\newtheorem{prop+}  [theo+]  {Proposition}
\newtheorem{coro+}  [theo+]  {Corollary}
\newtheorem{lemm+}  [theo+]  {Lemma}
\newtheorem{defi+}  [theo+]  {Definition}
\newtheorem{conj+}  [theo+]  {Conjecture}

\theoremstyle{definition}
\newtheorem{rema+}  [theo+]  {Remark}
\newtheorem{prob+}  [theo+]  {Problem}
\newtheorem{exam+}  [theo+]  {Example}
\newtheorem{ass+} [theo+] {Assumption}

\newenvironment{theorem}{\begin{theo+}}{\end{theo+}}
\newenvironment{proposition}{\begin{prop+}}{\end{prop+}}
\newenvironment{corollary}{\begin{coro+}}{\end{coro+}}
\newenvironment{lemma}{\begin{lemm+}}{\end{lemm+}}

\newcommand{\ti}{\mathrm i}

\newcommand{\id}{\operatorname{id}}
\newcommand{\supp}{\operatorname{supp}}

\newcommand{\hH}{{ H}}
\newcommand{\hP}{{ P}}
\newcommand{\hB}{{ B}}
\newcommand{\hS}{{ S}}

\allowdisplaybreaks[1]

\begin{document}

\baselineskip 18pt
\larger[2]
\title
[Higher order deformed elliptic Ruijsenaars operators] 
{Higher order deformed elliptic Ruijsenaars operators}
\author{Martin Halln\"as}
\address{Mathematical Sciences, Chalmers University of Technology and
G\"oteborg University, SE-412 96 G\"oteborg, Sweden}
\email{hallnas@chalmers.se}

\author{Edwin Langmann}
\address{Theoretical Physics,
KTH Royal Institute of Technology,
SE-106 91 Stockholm, Sweden}
\email{langmann@kth.se}

\author{Masatoshi Noumi}
\address{Department of Mathematics, KTH Royal Institute of Technology, SE-100 44 Stockholm, Sweden
(on leave from: Department of Mathematics, Kobe University, Rokko, Kobe 657-8501, Japan)}
\email{noumi@math.kobe-u.ac.jp}

\author{Hjalmar Rosengren}
\address{Mathematical Sciences, Chalmers University of Technology and
G\"oteborg University, SE-412 96 G\"oteborg, Sweden}
\email{hjalmar@chalmers.se}

%\subjclass[2020]{}
%\keywords{}

\begin{abstract}
We present four infinite families of mutually commuting difference operators  which include the deformed elliptic Ruijsenaars operators. The trigonometric limit of 
 this kind of operators was
previously  introduced by Feigin and Silantyev. 
They provide a quantum mechanical description of two kinds of relativistic quantum mechanical particles which can be identified with particles and anti-particles in an underlying quantum field theory. 
We give direct proofs of the commutativity of our operators and of some other fundamental properties such as
kernel function identities. In particular, we give a rigorous proof of the quantum integrability of the deformed Ruijsenaars model.
\end{abstract}

\maketitle        

\section{Introduction}

The quantum Calogero--Moser--Sutherland systems form an important class of
integrable systems in quantum mechanics. Chalykh, Feigin and Veselov \cite{cfv}
discovered that certain deformations of such systems maintain integrability. For
instance, the Schr\"odinger operator
\begin{align*} 
H &= -\sum_{i=1}^{m}\frac1{2}\frac{\partial^2}{\partial x_i^2} - g\sum_{i=1}^{r}\frac{1}{2}\frac{\partial^2}{\partial y_i^2} \\ 
&\quad + 
\sum_{1\leq i<j\leq m}\frac{g(g+1)}{(x_i-x_j)^2} + \sum_{1\leq i<j\leq r}\frac{1/g+1}{(y_i-y_j)^2}
 + \sum_{i=1}^m\sum_{j=1}^r \frac{g+1}{(x_i-y_j)^2}
\end{align*} 
is integrable for arbitrary variable numbers $m$ and $r$  and real coupling parameter $g$ \cite{cfv,se}, 
where $r=0$ is the non-deformed case first studied by Calogero \cite{c}. Such deformed models
turned out to be intimately connected to Lie superalgebras and related analogues
of symmetric functions such as super-Jack polynomials \cite{se,sv}. 
From a physics point of view, the deformed model describes a system of arbitrary numbers of two different kinds of identical particles. 
The Schr\"odinger operator above corresponds to the rational case; the most general elliptic case is obtained by replacing the interaction potential $1/x^2$ by
the Weierstrass $\wp$-function $\wp(x|\omega_1,\omega_2)$.

Ruijsenaars  \cite{r} introduced  relativistic generalizations of  quantum
Calogero--Moser--Sutherland systems, defined by difference operators rather than
  differential operators. 
  Deformed versions of such systems were first
considered by Chalykh \cite{ch1,ch2}. In greater generality, they were introduced and studied by 
Sergeev and Veselov \cite{sv,SV09b} in the trigonometric case and by Atai together with two of us \cite{ahl} in the  elliptic case. 
They describe systems of two kinds of identical particles which can be interpreted as particles and anti-particles in an underlying relativistic quantum field theory \cite{AHL20}. 
Feigin and Silantyev \cite{fs} constructed higher order operators that commute with the first order operators of Sergeev and Veselov.
They also showed that a sufficiently large subset of these operators is algebraically independent, concluding that the deformed models remain integrable in the 
relativistic setting.

In the present paper, we introduce and study elliptic extensions of the operators of Feigin and Silantyev. 
To be more precise, for fixed non-negative integers $m$ and $r$, we introduce a family
of operators  (see \eqref{docm} for the explicit expression)
\begin{equation}\label{fif}H_{m,r}^{(k)}(x_1,\dots,x_m;y_1,\dots,y_r;\delta,\kappa),\qquad k\in\mathbb Z_{\geq 0}.
\end{equation}
They are linear combinations of shift operators acting on 
functions in the $x$- and $y$-variables as
$$f(x_1,\dots,x_m;y_1,\dots,y_r)\mapsto f(x_1+\mu_1\delta,\dots,x_m+\mu_m\delta;y_1-I_1\kappa,\dotsm,y_r-I_r\kappa), $$
where $\mu_j\in\mathbb Z_{\geq 0}$, $I_j\in\{0,1\}$ and the total degree
$\sum_j\mu_j+\sum_j I_j$ is fixed to $k$. 
The parameters $\delta$ and $\kappa$ are related to the standard parameters
of Macdonald polynomial theory by $q=e^{2\ti\pi \delta}$, $t=e^{2\ti\pi\kappa}$. The
Ruijsenaars operators correspond to the case $m=0$ and the case $r=0$ give
the operators of Noumi and Sano \cite{ns}.

Up to a similarity transformation,
the Hamiltonian and momentum operator of the Ruijsenaars model are linear combinations of the operator $H_{0,r}^{(1)}$ and  the same operator with $(\delta,\kappa)$ replaced by $(-\delta,-\kappa)$. 
As we explain in Appendix \ref{appA}, a similar relation holds between the deformed Ruijsenaars model from \cite{ahl} and the operators  $H_{m,r}^{(1)}$. Thus, it is essential to consider commutation relations between  $H_{m,r}^{(k)}$ and appropriate modifications of these operators.
It turns out that there are four mutually commuting infinite families, given by \eqref{fif} together with 
\begin{subequations}\label{if}
\begin{align}
\label{sif}H_{r,m}^{(k)}(y_1,\dots,y_r;x_1,\dots,x_m;-\kappa,-\delta),
&\qquad k\in\mathbb Z_{\geq 0}, \\
\label{tif}H_{m,r}^{(k)}(x_1-\delta,\dots,x_m-\delta;y_1+\kappa,\dots,y_r+\kappa;-\delta,-\kappa), &\qquad k\in\mathbb Z_{\geq 0},\\
\label{foif}H_{r,m}^{(k)}(y_1+\kappa,\dots,y_r+\kappa;x_1-\delta,\dots,x_m-\delta;\kappa,\delta), &\qquad k\in\mathbb Z_{\geq 0}. \end{align}
\end{subequations}
Roughly speaking, in \eqref{sif} we have interchanged the roles of the $x$- and $y$-variables, in \eqref{tif} we have reversed the direction of the shift operators and in \eqref{foif} we have made both these changes. The shifts in the $x$- and $y$-variables present  in \eqref{tif} and \eqref{foif} could be eliminated by an overall translation (see \eqref{vs}), but we avoid that since it would make most of our formulas slightly more complicated.

The main result of the present paper is that the four infinite families of operators \eqref{fif} and \eqref{if}
mutually commute. Moreover, we prove that for generic $\delta$ and $\kappa$, the operators 
$H_{m,r}^{(k)}$
are 
algebraically independent for $1\leq k\leq m+r$. This gives a rigorous proof that 
 the deformed elliptic Ruijsenaars model is quantum integrable, which has until now been an unsolved problem. 
We also prove that the operators \eqref{sif} are in the algebraic closure of the 
operators \eqref{fif} (and vice versa). This generalizes the result of \cite{ns} that the Noumi--Sano operators are in the algebraic closure of the Ruijsenaars operators. 
The other two families are clearly outside this closure, since they act by shifts in the opposite direction.
Finally, we show that our operators satisfy kernel function identities with respect to the same kernel function that was obtained in \cite{ahl} in the first order case. 

Our proofs are direct and based  on non-trivial identities for theta functions that we refer to as source identities. They are also closely related to transformation formulas for multiple elliptic hypergeometric series found in \cite{kn,lsw,r2}.

In the main text, we present and prove the results in an additive notation close to the one used  by Ruijsenaars \cite{r}. For the convenience of the reader, in Appendix \ref{multapp} we give the key formulas in the multiplicative notation generalizing the one used in the theory of Macdonald polynomials \cite{Mac95}.

%In Appendix \ref{la} we explain the relation between the theta function identity underlying the commutativity of our operators and an elliptic hypergeometric transformation formula found by Langer, Schlosser and Warnaar \cite{lsw}.

{\bf Acknowledgements:} M.\ N.\ is grateful to the Knut and Alice Wallenberg Foundation for funding his
guest professorship at KTH. Financial support from the Swedish Research Council
is acknowledged by M.\ H.\ (Reg.\ nr.\ 2018-04291) and H.\ R.\ (Reg.\ nr.\ 2020-04221).

\section{Main results}

We fix a non-zero odd entire function  $x\mapsto[x]$, which satisfies the identity
\begin{equation}\label{wad}[x+y][x-y][u+v][u-v]+[x+v][x-v][y+u][y-u]=[x+u][x-u][y+v][y-v]. \end{equation}
A generic such function can be written 
\begin{equation}\label{ws}[x]=Ce^{c x^2}\sigma(x|\omega_1,\omega_2),\end{equation}
 where $\sigma$ is the Weierstrass sigma function \cite{ww}. 
 For our purposes, the prefactor $Ce^{cx^2}$ can be viewed as a normalization and plays no essential role. 
% Equivalently, one can take
% $$ [x]=C_2\, e^{c x^2}\theta_1(x/\omega_1;\omega_2/\omega_1),$$
% where $\theta_1$ is the first Jacobi theta function.
  Degenerate cases include the trigonometric solutions $[x]=\sin(\pi x/\omega)$, the hyperbolic solutions
  $[x]=\sinh(\pi x/\omega)$ and the rational solution $[x]=x$.  

Throughout,  $\delta$ and $\kappa$ are fixed parameters.
For simplicity, we will assume that 
\begin{equation}\label{gp}[n\delta]\neq 0,\quad [n\kappa]\neq 0,\qquad n\in\mathbb Z_{>0}. \end{equation}
See the end of Appendix \ref{appA} for a discussion of this condition.

For $k\in\mathbb Z_{\geq 0}$ we will write
\begin{equation}\label{sf}[x]_{k}=[x][x+\delta]\dotsm [x+(k-1)\delta] \end{equation}
and, for negative subscripts,
$$[x]_{-k}=\frac{1}{[x-k\delta ]_k}=\frac 1{[x-k\delta ][x-(k-1)\delta]\dotsm [x-\delta]}. $$
Occasionally, we indicate the dependence on $\delta$ as $[x;\delta]_k$. 
%{\blue To prove the Wronski relation, we will have
%$[x;\delta]_k$ and $[x;-\tau]_k$ appearing together.}
%We will also use the condensed notation
%$$[x_1,\dots,x_m]=[x_1]\dotsm [x_m], $$
%$$[x_1,\dots,x_m]_k=[x_1]_k\dotsm [x_m]_k. $$

We write $T_x^\delta$ for the difference operator
$$T_x^{\delta}f(x)=f(x+\delta)$$
and, more generally,
$$T_x^{\delta\mu}f(x_1,\dotsm,x_n)=f(x_1+\mu_1\delta,\dots,x_n+\mu_n\delta), $$
when $x$ and $\mu$ are vectors. 
We will write $\langle n\rangle=\{1,\dots,n\}$  
(the notation $[n]$ is more common, but we wish to avoid confusion with the 
function satisfying \eqref{wad}). 
Subsets $I\subseteq\langle n\rangle$ will be identified with
vectors $(I_1,\dots,I_n)\in\{0,1\}^n$, where $I_j=1$ for $j\in I$ and $I_j=0$ otherwise. With this identification, we can write
$T_x^{\delta I}=\prod_{j\in I}T_{x_j}^\delta$. 
%{\blue I prefer this to the notation with $\varepsilon_I$. It will be no problem to write things like $T_x^{\delta(I+J+\mu)}$ which will appear in the proofs.}
The higher order Ruijsenaars operators are   defined by 
\begin{equation}\label{ro}D_n^{(k)}=\sum_{I\subseteq\langle n\rangle,\,|I|=k}\,\prod_{i\in I,\,j\in I^c}\frac{[x_i-x_j+\kappa]}{[x_i-x_j]}\cdot T_x^{\delta I}, \end{equation}
where $I^c$ denotes the complement of $I$ in $\langle n\rangle$.
 It is a non-trivial fact that these operators commute for $0\leq k\leq n$ \cite{r}.

 Noumi and Sano \cite{ns} introduced another family of elliptic difference operators, which we denote 
$$H_n^{(k)}=\sum_{\mu\in\mathbb Z_{\geq 0}^n,\,|\mu|=k}\,
\prod_{1\leq i<j\leq n}\frac{[x_i-x_j+(\mu_i-\mu_j)\delta]}{[x_i-x_j]}
\prod_{i,j=1}^n\frac{[x_i-x_j+\kappa]_{\mu_i}}{[x_i-x_j+\delta]_{\mu_i}}\cdot
T_x^{\delta\mu}.$$
Here, $|\mu|=\mu_1+\dots+\mu_n$. They proved that they are related to the Ruijsenaars operators through the so called Wronski relation
 \begin{equation}\label{wr}\sum_{k+l=K} (-1)^k[k\kappa+l\delta] D_n^{(k)}H_n^{(l)}=0, \qquad K=1,2,3,\dots. \end{equation}
 This can be used to recursively write $H_n^{(l)}$ as a polynomial 
 in the operators $D_n^{(k)}$. As a consequence,
 all these operators commute.

 In the present work we introduce and study a family of difference operators $H_{m,r}^{(k)}$  in
 $m+r$ variables 
 $$(x_1,\dots,x_m,y_1,\dots,y_r),$$
  which generalize both $D_m^{(k)}$ and $H_r^{(k)}$. They are defined by
  \begin{subequations}\label{docm}
 \begin{equation}\label{do}H_{m,r}^{(k)}=\sum_{\substack{\mu\in\mathbb Z_{\geq 0}^m,\,I\subseteq \langle r\rangle,\\|\mu|+|I|=k}} 
 C_{\mu,I}(x;y)\,T_x^{\delta\mu}T_y^{-\kappa I},
 \end{equation}
where 
\begin{multline}\label{cm} C_{\mu,I}(x;y)=(-1)^{|I|}\prod_{1\leq i<j\leq m}\frac{[x_i-x_j+(\mu_i-\mu_j)\delta]}{[x_i-x_j]}
\prod_{i\in I,\,j\in I^c}\frac{[y_i-y_j-\delta]}{[y_i-y_j]}\\
\times\prod_{i,j=1}^m\frac{[x_i-x_j+\kappa]_{\mu_i}}{[x_i-x_j+\delta]_{\mu_i}}
\prod_{i=1}^m\left(\prod_{j\in I}\frac{[x_i-y_j-\kappa]}{[x_i-y_j+\mu_i\delta]}\prod_{j\in I^c}\frac{[x_i-y_j-\delta]}{[x_i-y_j+(\mu_i-1)\delta]}\right).
\end{multline}
\end{subequations}
If $m>0$, $H_{m,r}^{(k)}$ 
is well-defined only if $[j\delta]\neq 0$ for $1\leq j\leq k$, since
 otherwise the factor $[x_i-x_j+\delta]_{\mu_i}$ vanishes for $j=i$ and $\mu_i=k$. 
 This is guaranteed by our assumption \eqref{gp}.

Several special cases of the operators \eqref{do} are known in the literature. 
Clearly, $H_{m,0}^{(k)}$ is equal to the Noumi--Sano operator $H_m^{(k)}$.
The operator $H_{0,r}^{(k)}$ is equal to the Ruijsenaars operator $(-1)^kD_r^{(k)}$,  with 
 $\delta$ and $-\kappa$ interchanged. The operators $H_{m,r}^{(1)}$ are 
the deformed Ruijsenaars operators introduced in \cite{ahl}. 
Finally, the trigonometric limit of the general operators $H_{m,r}^{(k)}$ 
were introduced by Sergeev and Veselov \cite{sv} for $k=1$ and by  Feigin and Silantyev \cite{fs} in general; see also \cite{hlnr}.
To make the connection to the operators of \cite{fs} one should rewrite \eqref{cm} using
the elementary identity 
$$ \prod_{1\leq i<j\leq m}\frac{[x_i-x_j+(\mu_i-\mu_j)\delta]}{[x_i-x_j]}
\prod_{i,j=1}^m\frac{1}{[x_i-x_j+\delta]_{\mu_i}}
=(-1)^{|\mu|}\prod_{i,j=1}^m\frac{1}{[x_i-x_j-\delta\mu_j]_{\mu_i}}.
$$

Our first main result is that these operators commute.

\begin{theorem}\label{ct}
The operators \eqref{do} satisfy $[H_{m,r}^{(k)},H_{m,r}^{(l)}]=0$
for all $k,l\in\mathbb Z_{\geq 0}$.
\end{theorem}

%More precisely, this holds as an identity  in the algebra of formal difference operators with meromorphic coefficients.

Next, we  prove that $m+r$ of the operators $H_{m,r}^{(k)}$ are algebraically independent. We interpret this as a rigorous formulation of quantum integrability for
 the deformed elliptic Ruijsenaars model.

\begin{theorem}\label{ait}
For generic  $\kappa$ and $\delta$, the operators $H_{m,r}^{(k)}$, $k=1,\dots,m+r$, are algebraically independent.
\end{theorem}

As mentioned in the introduction, one can construct further commuting operators by making appropriate modification to $H_{m,r}^{(k)}$. We will first consider the family \eqref{tif}. Writing the coefficients \eqref{cm} as 
$C_{\mu,I}(x;y;\delta,\kappa)$, we denote the corresponding operators
\begin{equation}\label{hh}\hat H_{m,r}^{(k)}=\sum_{\substack{\mu\in\mathbb Z_{\geq 0}^m,\,I\subseteq \langle r\rangle,\\|\mu|+|I|=k}} 
 C_{\mu,I}(x_1-\delta,\dots,x_m-\delta;y_1+\kappa,\dots,y_r+\kappa;-\delta,-\kappa)\,T_x^{-\delta\mu}T_y^{\kappa I}.
 \end{equation}
 Since they are obtained from $H_{m,r}^{(k)}$ by a change of variables, these operators  commute among themselves. Our second main result states that the two families are mutually commuting.

\begin{theorem}\label{sct}
We have $[H_{m,r}^{(k)},\hat H_{m,r}^{(l)}]=0$ for all $k,l\in\mathbb Z_{\geq 0}$.
\end{theorem}

In the special case $m=0$, Theorem \ref{sct} follows  from Theorem \ref{ct} and the observation  that \cite{r}
$$\hat H_{0,r}^{(k)}=H_{0,r}^{(r-k)} \left(H_{0,r}^{(r)}\right)^{-1}. $$
 We stress that  when $m>0$ this simple argument does not work, and Theorem \ref{ct} requires a separate proof.

Next, we consider the family \eqref{sif}, which we denote
$$D_{m,r}^{(k)}=\sum_{\substack{\mu\in\mathbb Z_{\geq 0}^r,\,I\subseteq \langle m\rangle,\\|\mu|+|I|=k}} 
 C_{\mu,I}(y;x;-\kappa,-\delta)\,T_x^{\delta I}T_y^{-\kappa\mu}.$$
 The Ruijsenaars operator \eqref{ro} can be written $D_m^{(k)}=(-1)^kD_{m,0}^{(k)}$. 
 Our third main result states that the Wronski relation \eqref{wr} 
 extends to the deformed case $r>0$.

\begin{theorem}\label{wt}
The operators $D_{m,r}^{(k)}$ and $H_{m,r}^{(l)}$ are related by
\begin{equation}\label{dhr}\sum_{k+l=K}[k\kappa+l\delta] \,D_{m,r}^{(k)} H_{m,r}^{(l)}=0, \qquad K\in\mathbb Z_{>0}. \end{equation}
\end{theorem}

Since $D_{m,r}^{(0)}=\id$, we can alternatively write
\begin{equation}\label{wre}H_{m,r}^{(K)}=-\frac 1{[K\delta]}\sum_{k=1}^K[k\kappa+(K-k)\delta]\, D_{m,r}^{(k)} H_{m,r}^{(K-k)}. \end{equation}
This gives a recursion for computing $H_{m,r}^{(l)}$ as a polynomial in the operators $D_{m,r}^{(k)}$. As a consequence, we have the following result.

\begin{corollary}\label{wc}
The operator $ H_{m,r}^{(l)}$ is in the algebra generated by $D_{m,r}^{(k)}$ for $1\leq k\leq l$. In particular, $[D_{m,r}^{(k)},H_{m,r}^{(l)}]=0$ for
$k,\,l\in\mathbb Z_{\geq 0}$.
\end{corollary}

In \cite{ns}, the recursion \eqref{wre} for $r=0$ is solved explicitly in terms of determinants. This solution  extends immediately to general $r$.

\begin{corollary}\label{dc}
The operator $ H_{m,r}^{(l)}$ can be expressed in terms of the operators $D_{m,r}^{(k)}$ as
$$  H_{m,r}^{(l)}=(-1)^l\det_{1\leq i,j\leq l}\left(\frac{[(i-j+1)\kappa+(j-1)\delta]}{[i\delta]}\, D_{m,r}^{(i-j+1)}\right),$$
where matrix elements with $i-j+1<0$ are interpreted as zero.
\end{corollary}

Interchanging $\delta$ and $-\kappa$ gives the inverse relation
$$D_{m,r}^{(l)}=(-1)^l\det_{1\leq i,j\leq l}\left(\frac{[(i-j+1)\delta+(j-1)\kappa]}{[i\kappa]}\,H_{m,r}^{(i-j+1)}\right).$$
The  identities in \cite[Prop.\ 1.4]{ns} also extend immediately to our more general operators; we will not reproduce them here.

The fourth family of operators is
$$\hat D_{m,r}^{(k)}=\sum_{\substack{\mu\in\mathbb Z_{\geq 0}^r,\,I\subseteq \langle m\rangle,\\|\mu|+|I|=k}} 
 C_{\mu,I}(y_1+\kappa,\dots,y_r+\kappa;x_1-\delta,\dots,x_m-\delta;\kappa,\delta)\,T_x^{-\delta I}T_y^{\kappa\mu}.$$
It follows from Corollary \ref{wc} that $\hat D_{m,r}^{(k)}$ is a polynomial in
$\hat H_{m,r}^{(l)}$ for $l\leq k$. We can now conclude that all these operators commute.

\begin{corollary}\label{fcc}
For fixed $m$ and $r$, and arbitrary $k_j\in\mathbb Z_{\geq 0}$,
the operators $H_{m,r}^{(k_1)}$, $\hat H_{m,r}^{(k_2)}$, $D_{m,r}^{(k_3)}$ and $\hat D_{m,r}^{(k_4)}$ commute. 
\end{corollary}

Finally, we consider the so called kernel function. 
To this end,
we fix a meromorphic solution $G_\delta$ to the functional equation
\begin{equation}\label{gfe}G_\delta(x+\delta)=[x]G_\delta(x). \end{equation}
In the generic case, $G_\delta$ can be constructed from Ruijsenaars' elliptic gamma function, see \eqref{gd} below.

\begin{theorem}\label{kt}
Assuming that
\begin{equation}\label{kfc}(m-n)\kappa=(r-s)\delta,\end{equation}
the function 
\begin{multline}\label{ph}\Phi^{(m,r,n,s)}(x_1,\dots,x_m;y_1,\dots,y_r;X_1,\dots,X_n;Y_1,\dots,Y_s) \\
=\prod_{\substack{1\leq i\leq m,\\1\leq j\leq n}}\frac{G_\delta(x_i+X_j-\kappa)}{G_\delta(x_i+X_j)} 
\prod_{\substack{1\leq i\leq r,\\1\leq j\leq s}}\frac{G_{-\kappa}(y_i+Y_j+\delta)}{G_{-\kappa}(y_i+Y_j)}\\
\times\prod_{\substack{1\leq i\leq m,\\1\leq j\leq s}}[x_i+Y_j]\prod_{\substack{1\leq i\leq r,\\1\leq j\leq n}}[y_i+X_j]
\end{multline}
 satisfies the kernel function identity 
\begin{equation}\label{kfi}H_{m,r}^{(k)}(x;y)\Phi^{(m,r,n,s)}(x;y;X;Y)=H_{n,s}^{(k)}(X;Y)\Phi^{(m,r,n,s)}(x;y;X;Y), \end{equation}
where we indicate on which variables the difference operators act.
\end{theorem}

In particular, \eqref{kfi} holds when $m=n$ and $r=s$. 
However, also exceptional cases when $\kappa/\delta\in\mathbb Q$ may be of interest. 
The so called balancing condition \eqref{kfc} 
 stems from the fact that the sum of the zeroes of an elliptic function (modulo periods) equals the sum of the poles. This condition seems to be unavoidable in the elliptic case, but in the trigonometric case there is 
 a modified version of \eqref{kfi}   without this condition \cite{hlnr}.

\section{Source identities}

Clearly, the commutation relation in Theorem \ref{ct} can be  be translated to an identity involving the coefficients \eqref{cm}. This is also the case for Theorem \ref{sct}, Theorem \ref{wt} and Theorem \ref{kt}. We refer to these scalar equations for the coefficients as 
 \emph{source identities} for the corresponding facts about operators.  It turns out that 
 the operator identities  can be obtained from the
 same source identities as in the non-deformed case ($r=0$), but with  the variables
 specialized in a  non-obvious way.
 
 Theorem \ref{ct} and Theorem \ref{sct} will both be derived from the source identity \cite[Thm.\ A.2]{r}
\begin{subequations}\label{si}
\begin{multline}\label{rsi}\sum_{I\subseteq \langle n\rangle,\,|I|=k}\, \prod_{i\in I,\,j\in I^c}\frac{[z_i-z_j-a][z_i-z_j-b]}{[z_i-z_j][z_i-z_j-a-b]}\\
=\sum_{I\subseteq \langle n\rangle,\,|I|=n-k} \,\prod_{i\in I,\,j\in I^c}\frac{[z_i-z_j-a][z_i-z_j-b]}{[z_i-z_j][z_i-z_j-a-b]}.
\end{multline}
Ruijsenaars  used this identity to prove commutativity of the operators \eqref{ro}. 
In the case of
Theorem \ref{sct} we need to combine \eqref{si} with an argument of analytic continuation.
Incidentally, this  leads to a new proof of an elliptic hypergeometric transformation formula due to Langer, Schlosser and Warnaar \cite{lsw}.

For the Wronski relation \eqref{dhr}, the source identity is the same as the one used by Noumi and Sano \cite{ns} in the case $r=0$, that is,
\begin{equation}\label{nssi}\sum_{I\subseteq \langle n\rangle}(-1)^{|I|}\frac{[|z|-|w|+|I|a]}{[|z|-|w|]}
\prod_{i\in  I,\,j\in I^c}\frac{[z_i-z_j+a]}{[z_i-z_j]}\prod_{i\in I,\,j\in\langle n\rangle}
\frac{[z_i-w_j]}{[z_i-w_j+a]}=0.
 \end{equation}
 Here,  the notation $|z|=\sum_j z_j$ is used also for complex vectors.

Finally, the  kernel function identity \eqref{kfi} will be obtained from
the Kajihara--Noumi identity \cite{kn}
\begin{multline}\label{ksni}\sum_{I\subseteq\langle n\rangle,\,|I|=k}\prod_{i\in  I,\,j\in I^c}\frac{[z_i-z_j-a]}{[z_i-z_j]}\prod_{i\in I,\,j\in\langle n\rangle}
\frac{[z_i+w_j+a]}{[z_i+w_j]}\\
=\sum_{I\subseteq\langle n\rangle,\,|I|=k}\prod_{i\in  I,\,j\in I^c}\frac{[w_i-w_j-a]}{[w_i-w_j]}\prod_{i\in I,\,j\in\langle n\rangle}
\frac{[w_i+z_j+a]}{[w_i+z_j]}.
 \end{multline}
 The same identity was used by Ruijsenaars \cite{ru3} to prove the non-deformed case ($m=n=0$) of 
 \eqref{kfi}.
 Just as for Theorem \ref{sct}, it must in the general case be combined with an analytic continuation argument.

 \end{subequations}

Although the three source identities \eqref{si} may look similar at first glance,  none of them seem to follow easily from the others. 
Both
\eqref{nssi} and \eqref{ksni} can be derived as consequences of the  Frobenius determinant evaluation \cite{f}
$$\det_{1\leq i,j\leq n}\left(\frac{[\lambda+z_i+w_j]}{[\lambda][z_i+w_j]}\right)
=\frac{[\lambda+|z|+|w|]\prod_{1\leq i<j\leq n}[z_i-z_j][w_i-w_j]}{[\lambda]\prod_{1\leq i,j\leq n}[z_i+w_j]}.$$
However, we are not aware of an analogous proof of \eqref{rsi}.

\section{Commutativity}\label{cs}

In this section we prove Theorem \ref{ct}.
Consider  the product
\begin{equation}\label{hkl} H_{m,r}^{(k)}H_{m,r}^{(l)}=\sum_{\substack{\mu,\,\nu\in\mathbb Z_{\geq 0}^m,\,I,\,J\subseteq\langle r\rangle,\\|\mu|+|I|=k,\,|\nu|+|J|=l}}
C_{\mu,I}(x;y) C_{\nu,J}(x+\delta\mu;y-\kappa J) \,T_{x}^{\delta(\mu+\nu)}
T_y^{-\kappa(I+J)}.
\end{equation}
Here, $I+J$ means the sum of the corresponding sequences in $\{0,1\}^r$. 
It will be convenient to introduce the sets
$$ K=I\cap J,\quad L=I\triangle J,\quad M=\langle r\rangle\setminus (I\cup J), \quad
P=I\setminus J,\quad Q=J\setminus I,$$
where $\triangle$  denotes symmetric difference.
We then have the disjoint unions
$$\langle r\rangle=K\sqcup L\sqcup M, \qquad L=P\sqcup Q.$$
 Substituting $\nu\mapsto\lambda-\mu$, \eqref{hkl} takes the form
$$ H_{m,r}^{(k)}H_{m,r}^{(l)}=\sum_{\substack{\lambda\in\mathbb Z_{\geq 0}^m,\,K,\,L\subseteq \langle r\rangle,\\K\cap L=\emptyset,\,|\lambda|+2|K|+|L|=k+l}}
S_{k}(x;y)\,
 T_{x}^{\delta\lambda}
T_y^{-\kappa(2K+L)},
$$
where
\begin{equation}\label{sk}S_{k}(x;y)=\sum_{\substack{0\leq \mu_j\leq\lambda_j,\,1\leq j\leq m,\\
P\sqcup Q= L,\, |\mu|+|P|=k-|K|}}
C_{\mu,K\cup P}(x;y) C_{\lambda-\mu,K\cup Q}(x+\delta\mu;y-\kappa (K\cup P)).\end{equation}
Hence, the commutativity is equivalent to the symmetry
\begin{equation}\label{ss}S_{k}(x;y)=S_{|\lambda|+2|K|+|L|-k}(x;y), 
\end{equation}
for  fixed $\lambda$, $K$ and $L$.

We now insert the expression \eqref{cm} into  \eqref{sk}. Consider first the  factors involving only $y$-variables. They have the form
$$\prod_{i\in K\sqcup P,\,j\in M\sqcup Q}
\frac{[y_i-y_j-\delta]}{[y_i-y_j]} 
\prod_{t\in K\sqcup Q,\,u\in M\sqcup P}\frac{[y_t-y_u+\varepsilon_{t,u}\kappa-\delta]}{[y_t-y_u+\varepsilon_{t,u}\kappa]},
 $$
 where
 $$\varepsilon_{t,u}=(K\sqcup P)_u-(K\sqcup P)_t=\begin{cases} 1, & t\in Q,\,u\in P,\\
 -1, & t\in K, u\in M,\\ 
 0, & \text{else}.\end{cases} $$
The factors with $(i,j)\in P\times Q$ and $(t,u)\in Q\times P$ can be combined as
\begin{equation}\label{xfa}\prod_{i\in P,\,j\in Q}
\frac{[y_i-y_j-\delta][y_i-y_j-\kappa+\delta]}{[y_i-y_j][y_i-y_j-\kappa]}
\end{equation}
and
the remaining factors can be written
\begin{multline}\label{xfb} \prod_{i\in K,\,j\in M}
\frac{[y_i-y_j-\delta][y_i-y_j-\kappa-\delta]}{[y_i-y_j][y_i-y_j-\kappa]}\\
\times\prod_{i\in L,\,j\in M} \frac{[y_i-y_j-\delta]}{[y_i-y_j]}
\prod_{i\in K,\,j\in L} \frac{[y_i-y_j-\delta]}{[y_i-y_j]},
\end{multline}
For our purpose, the only relevant factors are \eqref{xfa}, since
\eqref{xfb} can be cancelled from \eqref{ss}.

The factors in \eqref{sk} involving only $x$-variables can be expressed as
\begin{multline*}\prod_{1\leq i<j\leq m}\frac{[x_i-x_j+(\mu_i-\mu_j)\delta]}{[x_i-x_j]}
\prod_{i,j=1}^m\frac{[x_i-x_j+\kappa]_{\mu_i}}{[x_i-x_j+\delta]_{\mu_i}} \\
\begin{split}&\quad\times
\prod_{1\leq i<j\leq m}\frac{[x_i-x_j+(\lambda_i-\lambda_j)\delta]}{[x_i-x_j+(\mu_i-\mu_j)\delta]}
\prod_{i,j=1}^m\frac{[x_i-x_j+(\mu_i-\mu_j)\delta+\kappa]_{\mu_i}}{[x_i-x_j+(\mu_i-\mu_j+1)\delta]_{\mu_i}}\\
&=\prod_{1\leq i<j\leq m}\frac{[x_i-x_j+(\lambda_i-\lambda_j)\delta]}{[x_i-x_j]}
\prod_{i,j=1}^m\frac{[x_i-x_j+\kappa]_{\lambda_i}}{[x_i-x_j+\delta]_{\lambda_i}}\\
&\quad\times\prod_{i,j=1}^m\frac{[x_i-x_j+\delta]_{\mu_i-\mu_j}[x_i-x_j+\kappa]_{\mu_i}[x_i-x_j-\lambda_j\delta]_{\mu_i}}{[x_i-x_j+\kappa]_{\mu_i-\mu_j}[x_i-x_j+\delta]_{\mu_i}[x_i-x_j-(\lambda_j-1)\delta-\kappa]_{\mu_i}},
\end{split}\end{multline*}
where the first two double products can be cancelled from \eqref{ss}.

Finally, the factors that mix $x$-variables and $y$-variables are
\begin{multline*}
\prod_{i=1}^m\left(\prod_{j\in K\sqcup P}\frac{[x_i-y_j-\kappa]}{[x_i-y_j+\mu_i\delta]}\prod_{j\in M\sqcup Q}\frac{[x_i-y_j-\delta]}{[x_i-y_j+(\mu_i-1)\delta]}\right)\\
\times \prod_{i=1}^m\left(\prod_{j\in K\sqcup Q}\frac{[x_i-y_j+\mu_i\delta-Q_j\kappa]}{[x_i-y_j+\lambda_i\delta+K_j\kappa]}\prod_{j\in M\sqcup P}\frac{[x_i-y_j+(\mu_i-1)\delta+P_j\kappa ]}{[x_i-y_j+(\lambda_i-1)\delta+P_j\kappa ]}\right).
\end{multline*}
Here, all factors with $j\in K\sqcup M$ can be cancelled from \eqref{ss}. The remaining factors  can be written
\begin{multline*} \prod_{i=1}^m\left(\prod_{j\in P}\frac{[x_i-y_j-\kappa][x_i-y_j+(\mu_i-1)\delta+\kappa ]}{[x_i-y_j+\mu_i\delta][x_i-y_j+(\lambda_i-1)\delta+\kappa ]}\right.\\
\times\left.\prod_{j\in Q}\frac{[x_i-y_j-\delta][x_i-y_j+\mu_i\delta-\kappa]}{[x_i-y_j+(\mu_i-1)\delta][x_i-y_j+\lambda_i\delta]}\right).\end{multline*}
From this expression, we factor out 
$$\prod_{i=1}^m\prod_{j\in P\sqcup Q}\frac{[x_i-y_j-\kappa]}{[x_i-y_j+\lambda_i\delta]}, $$
which can again be cancelled from \eqref{ss}, and are left with
\begin{multline*} \prod_{i=1}^m\left(\prod_{j\in P}\frac{[x_i-y_j+\lambda_i\delta][x_i-y_j+(\mu_i-1)\delta+\kappa ]}{[x_i-y_j+\mu_i\delta][x_i-y_j+(\lambda_i-1)\delta+\kappa ]}\right.\\
\times\left.\prod_{j\in Q}\frac{[x_i-y_j-\delta][x_i-y_j+\mu_i\delta-\kappa]}{[x_i-y_j+(\mu_i-1)\delta][x_i-y_j-\kappa]}\right).\end{multline*}

To summarize, to prove Theorem \ref{ct} it is enough to verify that \eqref{ss} holds with
\begin{align*}
S_k&=\sum_{\substack{0\leq \mu_j\leq\lambda_j,\,1\leq j\leq m,\\
P\sqcup Q= L,\, |\mu|+|P|=k-|K|}}\prod_{i\in P,\,j\in Q}
\frac{[y_i-y_j-\delta][y_i-y_j-\kappa+\delta]}{[y_i-y_j][y_i-y_j-\kappa]}\\
&\quad\times\prod_{i,j=1}^m\frac{[x_i-x_j+\delta]_{\mu_i-\mu_j}[x_i-x_j+\kappa]_{\mu_i}[x_i-x_j-\lambda_j\delta]_{\mu_i}}{[x_i-x_j+\kappa]_{\mu_i-\mu_j}[x_i-x_j+\delta]_{\mu_i}[x_i-x_j-(\lambda_j-1)\delta-\kappa]_{\mu_i}}\\
&\quad\times \prod_{i=1}^m\left(\prod_{j\in P}\frac{[x_i-y_j+\lambda_i\delta][x_i-y_j+(\mu_i-1)\delta+\kappa ]}{[x_i-y_j+\mu_i\delta][x_i-y_j+(\lambda_i-1)\delta+\kappa ]}\right.\\
&\quad\times\left.\prod_{j\in Q}\frac{[x_i-y_j-\delta][x_i-y_j+\mu_i\delta-\kappa]}{[x_i-y_j+(\mu_i-1)\delta][x_i-y_j-\kappa]}\right)
\end{align*}
(which differs from \eqref{sk} by a factor independent of $k$). 
It is enough to do this for $L=\langle r\rangle$, since the general case then follows by
changing the variables $\{y_1,\dots,y_r\}$ to 
 $\{y_j\}_{j\in L}$. We have thus reduced Theorem \ref{ct} to the following identity.

\begin{proposition}\label{csp}
For $\lambda\in\mathbb Z_{\geq 0}^m$, let
\begin{align*}
 S_k&=\sum_{\substack{0\leq\mu_j\leq \lambda_j,\,1\leq j\leq m\\P\subseteq \langle r\rangle,\,|\mu|+|P|=k}} \,
\prod_{i\in P,\,j\in P^c}\frac{[y_i-y_j-\delta][y_i-y_j+\delta-\kappa]}{[y_i-y_j][y_i-y_j-\kappa]}\\
 &\quad\times\prod_{i,j=1}^m\left(\frac{[x_i-x_j+\delta]_{\mu_i-\mu_j}}{[x_i-x_j+\kappa]_{\mu_i-\mu_j}}\frac{[x_i-x_j+\kappa]_{\mu_i}[x_i-x_j-\lambda_j\delta]_{\mu_i}}{[x_i-x_j+\delta]_{\mu_i}[x_i-x_j-(\lambda_j-1)\delta-\kappa]_{\mu_i}}\right)\\
 &\quad\times\prod_{i=1}^m\left(
\prod_{j\in P}\frac{[x_i-y_j+\lambda_i\delta][x_i-y_j+(\mu_i-1)\delta+\kappa]}{[x_i-y_j+\mu_i\delta][x_i-y_j+(\lambda_i-1)\delta+\kappa]}\right.\\
&\quad\times\left.\prod_{j\in P^c}\frac{[x_i-y_j-\delta][x_i-y_j+\mu_i\delta-\kappa]}{[x_i-y_j-\kappa][x_i-y_j+(\mu_i-1)\delta]}\right).
\end{align*}
Then, $S_k=S_{|\lambda|+r-k}$.
\end{proposition}

As we explain  in \S \ref{scs},
 Proposition \ref{csp} is closely related to an elliptic hypergeometric transformation formula due to Langer, Schlosser and Warnaar \cite{lsw}.

\begin{proof} 
Consider \eqref{rsi} with $a=\delta$, $b=\kappa-\delta$ and
 \begin{multline}\label{zs}(z_1,\dots,z_n)
 = \big(x_1,x_1+\delta,\dots,x_1+(\lambda_1-1)\delta,\dots,\\
x_m,x_m+\delta,\dots,x_m+(\lambda_m-1)\delta,
y_1,\dots,y_r\big),
 \end{multline}
 where $n=|\lambda|+r$. 
 As usual, we identify sets $I\subseteq \langle n\rangle$ with sequences in $\{0,1\}^n$. 
 We claim that, to give
 a non-zero contribution to the sums in \eqref{rsi},  $I$ has to be of the form
 \begin{equation*}I=(\underbrace{1,\dots,1}_{\mu_1},\underbrace{0,\dots,0}_{\lambda_1-\mu_1},\dots,\underbrace{1,\dots,1}_{\mu_m},\underbrace{0,\dots,0}_{\lambda_m-\mu_m},P), \end{equation*}
 where $0\leq\mu_k\leq \lambda_k$ for each $k$ and $P\subseteq \langle r\rangle$. 
 Otherwise,  there is an index $k$ such that
  $z_{k+1}=z_k+\delta$, 
 $k\notin I$ and $k+1\in I$. Then, the corresponding term in \eqref{rsi} contains the factor $[z_{k+1}-z_{k}-\delta]=0$. 
 
 The general term  in \eqref{rsi} can be written
 $F(\delta)/F(\kappa)$,
 where
 $$F(c)=\prod_{i\in I,\,j\in I^c}\frac{[z_i-z_j-c]}{[z_i-z_j-c+\delta]}. $$
Specializing $z$ as in \eqref{zs}, $F(c)$ splits naturally  into four parts, depending on whether $z_i$ and $z_j$ are specialized to shifted $x$-variables or to $y$-variables. 
The first part is 
\begin{align*}F_1(c)&=\prod_{i,j=1}^m\prod_{\substack{1\leq k\leq \mu_i,\\ \mu_{j}+1\leq l\leq\lambda_j}}\frac{[x_i-x_j+(k-l)\delta-c]}{[x_i-x_j+(k-l+1)\delta-c]}\\
&=\prod_{i,j=1}^m\prod_{k=1}^{\mu_i}\frac{[x_i-x_j+(k-\lambda_j)\delta-c]}{[x_i-x_j+(k-\mu_j)\delta-c]}=\prod_{i,j=1}^m\frac{[x_i-x_j+(1-\lambda_j)\delta-c]_{\mu_i}}{[x_i-x_j+(1-\mu_j)\delta-c]_{\mu_i}},
 \end{align*}
 where we used that the product in $l$ telescopes. Using that
 $$[x_i-x_j+(1-\mu_j)\delta-c]_{\mu_i}
=(-1)^{\mu_i}\frac{[x_j-x_i+c]_{\mu_j}}{[x_j-x_i+c]_{\mu_j-\mu_i}},
 $$
 we obtain
 \begin{subequations}\label{fi}
 \begin{equation}F_1(c)=(-1)^{m|\mu|} \prod_{i,j=1}^m\frac{[x_i-x_j+c]_{\mu_i-\mu_j}[x_i-x_j+(1-\lambda_j)\delta-c]_{\mu_i}}{[x_i-x_j+c]_{\mu_i}}.\end{equation}
 
 The second part of the product, when $z_i$ is specialized to a shifted $x$-variable and $z_j$ to a $y$-variable, can be written
 \begin{equation}F_2(c)=\prod_{i=1}^m\prod_{j\in P^c}\prod_{k=1}^{\mu_i}
 \frac{[x_i+(k-1)\delta-y_j-c]}{[x_i+k\delta-y_j-c]}=\prod_{i=1}^m\prod_{j\in P^c}
 \frac{[x_i-y_j-c]}{[x_i-y_j+\mu_i\delta-c]}
  \end{equation}
and  similarly the third part is
\begin{equation}F_3(c)=\prod_{i=1}^m\prod_{j\in P}\prod_{k=\mu_i+1}^{\lambda_i}
 \frac{[x_i+(k-1)\delta-y_j+c]}{[x_i+(k-2)\delta-y_j+c]}=\prod_{i=1}^m\prod_{j\in P}
 \frac{[x_i-y_j+(\lambda_i-1)\delta+c]}{[x_i-y_j+(\mu_i-1)\delta+c]}. \end{equation}
Finally, the last part is simply
\begin{equation}F_4(c)=\prod_{i\in P,\,j\in P^c}\frac{[y_i-y_j-c]}{[y_i-y_j+\delta-c]}. 
\end{equation}
\end{subequations}

The general term of the sums in \eqref{rsi} is
$$\frac{F_1(\delta)F_2(\delta)F_3(\delta)F_4(\delta)}{F_1(\kappa)F_2(\kappa)F_3(\kappa)F_4(\kappa)}. $$
Inserting the explicit expressions \eqref{fi} yields  the desired result.
\end{proof}

\section{Algebraic independence}

We will now prove Theorem \ref{ait}. We will first 
give a proof of algebraic independence in the special case $\kappa=\delta$, and then deduce the general case. 

\begin{lemma}\label{sil}
For $\kappa=\delta$, the operators $H_{m,r}^{(k)}$, $k=1,\dots,m+r$, are algebraically independent.
\end{lemma}

\begin{proof}
It is easy to check that
$$ C_{\mu,I}(x;y)\Big|_{\kappa=\delta}=(-1)^{|I|}F(x;y)^{-1} T_x^{\delta\mu} T_y^{-\delta I} F(x;y),$$
where
$$F(x;y)=\frac{\prod_{1\leq i<j\leq m}[x_i-x_j]\prod_{1\leq i<j\leq r}[y_i-y_j]}{\prod_{1\leq i\leq m,\,1\leq j\leq r}[x_i-y_j-\delta]}. $$
It follows that
$$H_{m,r}^{(k)}(x;y)\Big|_{\kappa=\delta}=F(x;y)^{-1}\sum_{\substack{\mu\in\mathbb Z_{\geq 0}^m,\,I\subseteq \langle r\rangle,\\|\mu|+|I|=k}} 
(-1)^{|I|} T_x^{\delta\mu} T_y^{-\delta I} F(x;y).$$
Thus, it is enough to prove algebraic independence of the  operators
$$\sum_{\substack{\mu\in\mathbb Z_{\geq 0}^m,\,I\subseteq \langle r\rangle,\\|\mu|+|I|=k}} 
(-1)^{|I|} T_x^{\delta\mu} T_y^{-\delta I}
=\sum_{\substack{i\geq 0,\,0\leq j\leq r,\\ i+j=k}}(-1)^j h_i(T_x^{\delta})
e_j(T_y^{-\delta}),\qquad 1\leq k\leq m+r,
$$
 where $h_i$ and $e_j$ denote homogeneous and elementary  symmetric polynomials, respectively. 
This is in turn equivalent to algebraic independence of the polynomials
$$h_{m,r}^{(k)}(\xi;\eta)=\sum_{\substack{i\geq 0,\,0\leq j\leq r,\\ i+j=k}}(-1)^j h_i(\xi)
e_j(\eta),\qquad 1\leq k\leq m+r, $$
which belong to the ring 
$$\mathbb C[\xi_1,\dots,\xi_m,\eta_1,\dots,\eta_r]^{\mathrm S_m\times\mathrm S_r}
=\mathbb C[e_1(\xi),\dots,e_m(\xi),e_1(\eta),\dots,e_r(\eta)].$$
We introduce the total order on the monomials 
$$e_r(\eta)^{l_r}\dotsm e_1(\eta)^{l_1}e_m(\xi)^{k_m}\dotsm e_1(\xi)^{k_1} $$
corresponding to lexicographic order of the multi-indices $(l_r,\dots,l_1,k_m,\dots,k_1)$.

It is well-known that
$$h_i(\xi)=(-1)^{i-1}e_i(\xi)+\text{lower order terms},\qquad 1\leq i\leq m. $$
Hence,  for $1\leq k\leq r$ we have
$$h_{m,r}^{(k)}(\xi;\eta)=(-1)^k e_k(\eta)+\text{lower order terms}, $$
whereas for $r+1\leq k\leq m+r$,
\begin{align*}h_{m,r}^{(k)}(\xi;\eta)&=(-1)^r h_{k-r}(\xi) e_r(\eta)+\text{lower order terms}\\
&=(-1)^{k+1} e_{k-r}(\xi)e_r(\eta)+\text{lower order terms}.
 \end{align*}
It follows that  the polynomial
\begin{equation}\label{pih}h_{m,r}^{(1)}(\xi;\eta)^{\lambda_1}\dotsm  h_{m,r}^{(m+r)}(\xi;\eta)^{\lambda_{m+r}},\qquad \lambda_1,\dots,\lambda_{m+r}\in\mathbb Z_{\geq 0}
\end{equation} 
has leading term 
$$ e_1(\eta)^{\lambda_1}\dotsm e_{r-1}(\eta)^{\lambda_{r-1}} e_r(\eta)^{\lambda_r+\lambda_{r+1}+\dots+\lambda_{m+r}}e_{1}(\xi)^{\lambda_{r+1}}\dotsm e_m(\xi)^{\lambda_{m+r}}, 
$$
up to an irrelevant sign factor.  Since these terms are all distinct, the polynomials
\eqref{pih} are linearly independent. Equivalently, $h_{m,r}^{(k)}$ are algebraically independent for $1\leq k\leq m+r$. 
 \end{proof}

We will now prove Theorem \ref{ait}. Algebraic independence is equivalent to linear independence of the operators
\begin{equation}\label{mh}H_{m,r}^\lambda=(H_{m,r}^{(1)})^{\lambda_1}\dotsm (H_{m,r}^{(m+r)})^{\lambda_{m+r}},\qquad \lambda\in\mathbb Z_{\geq 0}^{m+r}.
\end{equation}
These operators have the  form
$$H_{m,r}^\lambda=\sum_{|\mu|+|\nu|=\|\lambda\|}C_{\mu,\nu}^\lambda(x;y)
T_x^{\delta\mu}T_y^{-\kappa\nu},$$
where the coefficients $C_{\mu,\nu}^\lambda$ are  meromorphic and
$$\|\lambda\|=\lambda_1+2\lambda_2+\dots+(m+r)\lambda_{m+r}. $$
A linear relation
$$\sum_\lambda c_\lambda H_{m,r}^\lambda=0 $$
between the operators \eqref{mh} is equivalent to the corresponding relations
$$\sum_{\|\lambda\|=|\mu|+|\nu|} c_\lambda C_{\mu,\nu}^\lambda(x;y)=0,\qquad 
\mu\in\mathbb Z_{\geq 0}^m,\quad \nu\in\mathbb Z_{\geq 0}^r $$
between the coefficients. In particular, the operators are linearly independent if and only if the matrices 
$$ ( C_{\mu,\nu}^\lambda(x;y))_{|\mu|+|\nu|=N,\,\|\lambda\|=N}$$
of meromorphic functions
have maximal rank for each $N\in\mathbb Z_{\geq 0}$. 
Since we know from Lemma \ref{sil} that this is the case when $\kappa=\delta$, it must  be true for generic values of $\kappa$ and $\delta$.

\section{Second commutation relation}
\label{scs}

The proof of Theorem \ref{sct} is similar to that of Theorem \ref{ct}. We write  \eqref{hh}   as
$$ 
\hat H_{m,r}^{(k)}=\sum_{\substack{\mu\in\mathbb Z_{\geq 0}^m,\,I\subseteq \langle r\rangle,\\|\mu|+|I|=k}} 
 D_{\mu,I}(x;y)\,T_x^{-\delta\mu}T_y^{\kappa I},
$$
where
\begin{multline}\label{dm} D_{\mu,I}(x;y)\\
=
(-1)^{|I|}\prod_{1\leq i<j\leq m}\frac{[x_i-x_j-(\mu_i-\mu_j)\delta]}{[x_i-x_j]}
\prod_{i,j=1}^m\frac{[x_j-x_i+\kappa]_{\mu_i}}{[x_j-x_i+\delta]_{\mu_i}}
\prod_{i\in I,\,j\in I^c}\frac{[y_i-y_j+\delta]}{[y_i-y_j]}
\\
\times
\prod_{i=1}^m\left(\prod_{j\in I}\frac{[x_i-y_j-\delta]}{[x_i-y_j-(\mu_i+1)\delta-\kappa]}\prod_{j\in I^c}\frac{[x_i-y_j-\kappa]}{[x_i-y_j-\mu_i\delta-\kappa]}\right).
\end{multline}
This gives
\begin{multline*}H_{m,r}^{(k)}\hat H_{m,r}^{(l)}\\
=\sum_{\substack{\lambda\in\mathbb Z^m,\,K\in\{-1,0,1\}^r,\\|\lambda|+|K|=k-l}}\,\sum_{\substack{\mu\in\mathbb Z_{\geq 0}^m,\,I\subseteq \langle r\rangle,\\|\mu|+|I|=k}}
   C_{\mu,I}(x;y)D_{\mu-\lambda,I-K}(x+\delta\mu;y-\kappa I)
\,T_x^{\delta\lambda}T_y^{-\kappa K}.\end{multline*}
 Here, we should interpret $D_{\mu-\lambda,I-K}$ as zero unless
   $\mu_j\geq\lambda_j$ and $I_j-K_j\in\{0,1\}$ for all $j$.
In the notation
$$K^i=\{j\in\langle r\rangle;\,K_j=i\},\qquad i=-1,\,0,\,1, $$
the latter condition is equivalent to
$K^1\subseteq I\subseteq K^0\cup K^1$. Then, $I-K=K^{-1}\cup (I\cap K^0)$.
Writing $\hat H_{m,r}^{(l)}H_{m,r}^{(k)}$ in the same way we find that Theorem \ref{sct} is equivalent to the scalar equations
\begin{multline}\label{hdhs}\sum_{\substack{\mu_j\geq\max(0,\lambda_j),\,1\leq j\leq m,\\
K^1\subseteq I\subseteq K^0\cup K^1,\\
|\mu|+|I|=k}} C_{\mu,I}(x;y)D_{\mu-\lambda,I-K}(x+\delta\mu;y-\kappa I)\\
=\sum_{\substack{\mu_j\geq\max(0,\lambda_j),\,1\leq j\leq m,\\
K^1\subseteq I\subseteq K^0\cup K^1,\\
|\mu|+|I|=k}} C_{\mu,I}(x-\delta(\mu-\lambda);y+\kappa(I-K))D_{\mu-\lambda,I-K}(x;y).
 \end{multline}

We want to factor
$$  C_{\mu,I}(x;y)D_{\mu-\lambda,I-K}(x+\delta\mu;y-\kappa I)=F(x;y)G(x;y),$$
where $F$  is independent of $\mu$ and $I$, and $G$ is normalized to take the value $1$ if $\mu=0$ and $I=K^1$. 
Inserting the explicit expressions  \eqref{cm} and \eqref{dm}, we find after a  tedious  computation that
\begin{align*}F(x;y) &=(-1)^{|K|}\prod_{1\leq i<j\leq m}\frac{[x_i-x_j+(\lambda_i-\lambda_j)\delta]}{[x_i-x_j]}\prod_{i\in K^0,\,j\in K^{-1}}\frac{[y_i-y_j-\delta]}{[y_i-y_j]}\\
&\quad\times\prod_{i\in K^1,\,j\in K^{0}}\frac{[y_i-y_j-\delta]}{[y_i-y_j]}
\prod_{i\in K^1,\,j\in K^{-1}}\frac{[y_i-y_j-\delta][y_i-y_j-\delta-\kappa]}{[y_i-y_j][y_i-y_j-\kappa]}
\\
&\quad\times
\prod_{i=1}^m
\left(
\prod_{j\in K^{-1}}\frac{[x_i-y_j-\delta]}{[x_i-y_j+(\lambda_i-1)\delta-\kappa]}
\prod_{j\in K^0}\frac{[x_i-y_j-\kappa]}{[x_i-y_j+\lambda_i\delta-\kappa]}\right.\\
&\quad\quad\times\left.\prod_{j\in K^1}\frac{[x_i-y_j-\kappa]}{[x_i-y_j+\lambda_i\delta]}
\right),
\end{align*}
\begin{align*}
G(x;y)
&=\prod_{i\in (K^0\cap I),\,j\in (K^0\setminus I)}\frac{[y_i-y_j-\delta][y_i-y_j+\delta-\kappa]}{[y_i-y_j][y_i-y_j-\kappa]}\\
&\quad\times\prod_{i,j=1}^m\left(\frac{[x_i-x_j+\delta]_{\mu_i-\mu_j}}{[x_i-x_j+\kappa]_{\mu_i-\mu_j}}\frac{[x_i-x_j+\kappa]_{\mu_i}[x_i-x_j+\kappa]_{\mu_i-\lambda_j}}{[x_i-x_j+\delta]_{\mu_i}[x_i-x_j+\delta]_{\mu_i-\lambda_j}}\right)\\
 &\quad\times\prod_{i=1}^m\left(
\prod_{j\in K^0\cap I}\frac{[x_i-y_j+\lambda_i\delta-\kappa][x_i-y_j+(\mu_i-1)\delta+\kappa]}{[x_i-y_j+(\lambda_i-1)\delta][x_i-y_j+\mu_i\delta]}\right.\\
&\quad\quad\times\left.\prod_{j\in K^0\setminus I}\frac{[x_i-y_j-\delta][x_i-y_j+\mu_i\delta-\kappa]}{[x_i-y_j-\kappa][x_i-y_j+(\mu_i-1)\delta]}\right).
\end{align*}
In the same way, one  gets
\begin{multline*}C_{\mu,I}(x-\delta(\mu-\lambda);y+\kappa(I-K))D_{\mu-\lambda,I-K}(x;y)\\
=\prod_{i,j=1}^m\frac{[x_i-x_j+\kappa]_{\lambda_i-\lambda_j}}{[x_i-x_j+\delta]_{\lambda_i-\lambda_j}}\,F(x;y)G(\hat x;\hat y),
 \end{multline*}
where we have introduced the variables
\begin{equation}\label{hyz}\hat x_j=-x_j-\lambda_j\delta+\kappa,\qquad \hat y_j=-y_j-\delta. \end{equation}

We now observe that the variables $y_j$ with $j\in K^{-1}\cup K^1$ only appear in the prefactor $F$ that can be cancelled from \eqref{hdhs}. Thus, it suffices to prove the case $K^0=\langle r\rangle$, that is, the identity 
\begin{multline}\label{hdha}\sum_{\substack{\mu_j\geq\max(0,\lambda_j),\,1\leq j\leq m,\\
 I\subseteq \langle r\rangle,\,
|\mu|+|I|=k}} G(x;y)\\
=\prod_{i,j=1}^m\frac{[x_i-x_j+\kappa]_{\lambda_i-\lambda_j}}{[x_i-x_j+\delta]_{\lambda_i-\lambda_j}}\sum_{\substack{\mu_j\geq\max(0,\lambda_j),\,1\leq j\leq m,\\
 I\subseteq \langle r\rangle,\,
|\mu|+|I|=k}} G(\hat x;\hat y).
\end{multline}

We will identify \eqref{hdha} with a version of 
 Proposition \ref{csp} where the conditions $0\leq\mu_j\leq\lambda_j$ are replaced by $\mu_j\geq \max(0,\lambda_j)$. To this end, we first rewrite the identity
 $S_k=S_{|\lambda|+r-k}$.
On the left-hand side, we note that the  factor $[x_i-x_j-\lambda_j\delta]_{\mu_i}$ vanishes if $j=i$ and $\mu_i>\lambda_i$. Hence, we may ignore the restrictions $\mu_i\leq \lambda_i$ on the summation indices. It will be convenient to introduce the variables $z_j=-x_j-\lambda_j\delta$. We can then write
$$S_k=T_k(x;y;-x-\delta\lambda), $$ 
where
\begin{align*}
 &\quad T_k(x_1,\dots,x_m;y_1,\dots,y_r;z_1,\dots,z_m)\\
 &=\sum_{\substack{\mu\in\mathbb Z_{\geq 0}^m,\,P\subseteq \langle r\rangle,\\|\mu|+|P|=k}} \,
\prod_{i\in P,\,j\in P^c}\frac{[y_i-y_j-\delta][y_i-y_j+\delta-\kappa]}{[y_i-y_j][y_i-y_j-\kappa]}\\
 &\quad\times\prod_{i,j=1}^m\left(\frac{[x_i-x_j+\delta]_{\mu_i-\mu_j}}{[x_i-x_j+\kappa]_{\mu_i-\mu_j}}\frac{[x_i-x_j+\kappa]_{\mu_i}[x_i+z_j]_{\mu_i}}{[x_i-x_j+\delta]_{\mu_i}[x_i+z_j+\delta-\kappa]_{\mu_i}}\right)\\
 &\quad\times\prod_{i=1}^m\left(
\prod_{j\in P}\frac{[z_i+y_j][x_i-y_j+(\mu_i-1)\delta+\kappa]}{[z_i+y_j+\delta-\kappa][x_i-y_j+\mu_i\delta]}\right.\\
&\quad\times\left.\prod_{j\in P^c}\frac{[x_i-y_j-\delta][x_i-y_j+\mu_i\delta-\kappa]}{[x_i-y_j-\kappa][x_i-y_j+(\mu_i-1)\delta]}\right).
\end{align*}

In the sum $S_{|\lambda|+r-k}$, we  replace $\mu_i\mapsto\lambda_i-\mu_i$ for all $i$ and $P\mapsto P^c$. By a straight-forward computation, we obtain
$$S_{|\lambda|+r-k}=T_k (-x-\delta\lambda;\hat y;x),$$
where $\hat y$ is as in \eqref{hyz}.
 Thus, Proposition \ref{csp} can be formulated as
\begin{equation}\label{tt}T_k(x;y;z)
=T_k(z;\hat y;x),\end{equation} 
 where
 \begin{equation}\label{xzr}x_j+z_j=-\lambda_j\delta,\qquad \lambda_j\in\mathbb Z_{\geq 0},\quad j=1,\dots,m. \end{equation}

\begin{proposition}\label{lsl}
The identity \eqref{tt} holds also without the condition \eqref{xzr}.
\end{proposition} 
 
 \begin{proof}
 We apply a standard analytic continuation argument, see e.g.\ \cite{w}. It is straight-forward to check that each term in \eqref{tt} has the form
 $$C\prod_{j=1}^N\frac{[x_1+a_j]}{[x_1+b_j]}, $$
 where $C$ is independent of $x_1$ and
 $$\sum_{j=1}^N(a_j-b_j)=k(\kappa-\delta). $$
 Assume that we are in a generic situation, when $[x]$ is given by \eqref{ws}.
 Then, all these terms have the same quasi-periodicity with respect to the lattice $\Gamma=\mathbb Z\omega_1+\mathbb Z\omega_2$. It follows that \eqref{tt} holds 
 if 
 $$x_1\in -z_1-\mathbb Z_{\geq 0}\delta+\Gamma.$$
 By our assumption \eqref{gp}, these values are all distinct mod $\Gamma$, so by analytic continuation \eqref{tt} holds for generic $x_1$. By symmetry, the same argument applies to the other variables $x_j$.
 \end{proof}

In the special case $r=0$, Proposition \ref{lsl} reduces to the elliptic hypergeometric transformation formula
\begin{multline}\label{lss}
\sum_{\mu\in\mathbb Z_{\geq 0}^m,\,|\mu|=k}\, \prod_{i,j=1}^m\left(\frac{[x_i-x_j+\delta]_{\mu_i-\mu_j}}{[x_i-x_j+\kappa]_{\mu_i-\mu_j}}\frac{[x_i-x_j+\kappa]_{\mu_i}[x_i+z_j]_{\mu_i}}{[x_i-x_j+\delta]_{\mu_i}[x_i+z_j+\delta-\kappa]_{\mu_i}}\right)\\
=\sum_{\mu\in\mathbb Z_{\geq 0}^m,\,|\mu|=k}\, \prod_{i,j=1}^m\left(\frac{[z_i-z_j+\delta]_{\mu_i-\mu_j}}{[z_i-z_j+\kappa]_{\mu_i-\mu_j}}\frac{[z_i-z_j+\kappa]_{\mu_i}[z_i+x_j]_{\mu_i}}{[z_i-z_j+\delta]_{\mu_i}[z_i+x_j+\delta-\kappa]_{\mu_i}}\right).
\end{multline}
Replacing $m$ by $m+1$ and eliminating the summation index $\mu_{m+1}$, it is straight-forward to check that this is equivalent to \cite[Cor. 4.3]{lsw}. 
Conversely, one can recover Proposition \ref{lsl} from \eqref{lss} by replacing $m$ by $m+r$ and then specializing $x_j+z_j=-\delta$ for $m+1\leq j\leq m+r$. 
The proof of \eqref{lss} given here is very similar to that in \cite{lsw}. However, the observation that \eqref{lss} can be derived from Ruijsenaars' identity \eqref{rsi} is new. 
In \cite{lsw}, it is derived from a more complicated source identity.

We will now consider \eqref{tt} when $z_j=\hat x_j$ is given by \eqref{hyz}.
More precisely, to avoid division by zero we first multiply both sides with
$$\prod_{i,j=1}^n\frac{[x_i+z_j+\delta-\kappa]_{\lambda_j}}{[x_i+z_j]_{\lambda_j}} $$
and then use that
$$\lim_{z_j\rightarrow -x_j-\lambda_j\delta+\kappa}\frac{[x_i+z_j+\delta-\kappa]_{\lambda_j}[x_i+z_j]_{\mu_i}}{[x_i+z_j]_{\lambda_j}[x_i+z_j+\delta-\kappa]_{\mu_i}}=\frac{[x_i-x_j+\kappa]_{\mu_i-\lambda_j}}{[x_i-x_j+\delta]_{\mu_i-\lambda_j}}, $$
which by definition vanishes if $j=i$ and $\mu_i<\lambda_i$.
On the right-hand side, we use
$$\lim_{x_j\rightarrow -z_j-\lambda_j\delta+\kappa}\frac{[z_i+x_j+\delta-\kappa]_{\lambda_i}[z_i+x_j]_{\mu_i}}{[z_i+x_j]_{\lambda_i}[z_i+x_j+\delta-\kappa]_{\mu_i}}=\frac{[x_j-x_i+\kappa]_{\lambda_j-\lambda_i}[z_i-z_j+\kappa]_{\mu_i-\lambda_j}}{[x_j-x_i+\delta]_{\lambda_j-\lambda_i}[z_i-z_j+\delta]_{\mu_i-\lambda_j}}. $$
It is now clear that the resulting limit case of \eqref{tt} is \eqref{hdha} (with $I$=$P$). This completes the proof of Theorem \ref{sct}.

\section{Wronski relations}

We now turn to the proof of Theorem \ref{wt}. Indicating the parameter-dependence in 
 \eqref{cm} as $ C_{\mu,I}(x;y;\delta,\kappa)$, the left-hand side of \eqref{wr} may be expressed as
 \begin{multline*}\sum_{\substack{\mu\in\mathbb Z_{\geq 0}^r,\,I\subseteq \langle m\rangle,\\ |\nu|\in\mathbb Z_{\geq 0}^m,\,J\subseteq\langle r\rangle,\\
 |\mu|+|I|+|\nu|+|J|=K}}
 [(|\mu|+|I|)\kappa+(|\nu|+|J|)\delta]\\
\times C_{\mu,I}(y;x;-\kappa,-\delta)C_{\nu,J}(x+\delta I;y-\kappa\mu;\delta,\kappa)\,T_x^{\delta(\nu+I)}T_y^{-\kappa(J+\mu)}.
\end{multline*}
We make the change of variables
 $\nu\mapsto\nu-I$ and $\mu\mapsto\mu-J$. We then have $\nu_j\geq I_j$ and $\mu_j\geq J_j$ for all $j$, that is, 
 $$I\subseteq \supp(\nu)=\{j\in\langle m\rangle;\,\nu_j>0\}$$
  and $J\subseteq\supp(\mu)$. 
 This gives the expression
 \begin{multline*}
\sum_{\substack{\mu\in\mathbb Z_{\geq 0}^r,\,\nu\in\mathbb Z_{\geq 0}^m,\\
 |\mu|+|\nu|=K}}\,
 \bigg(\sum_{\substack{I\subseteq \supp(\nu),\\J\subseteq \supp(\mu)}} 
 [(|\mu|+|I|-|J|)\kappa+(|\nu|+|J|-|I|)\delta]\\
\times C_{\mu-J,I}(y;x;-\kappa,-\delta)C_{\nu-I,J}(x+\delta I;y+\kappa(J-\mu);\delta,\kappa)\bigg) T_x^{\delta\nu}T_y^{-\kappa \mu}. \end{multline*}
 We introduce the notation $M=\supp(\nu)$ and $N=\supp(\mu)$, and normalize the inner sum so  that the term with $I=\emptyset$ and $J=N$ is $1$. That is, we define
   \begin{align}\nonumber D_{\mu,\nu}(x;y)&=\sum_{I\subseteq M,\,J\subseteq N} \frac{[(\mu+|I|-|J|)\kappa+(|\nu|+|J|-|I|)\delta]}{[(|\mu|-|N|)\kappa+(|\nu|+|N|)\delta]}\\
\label{dmn} &\quad  \times\frac{C_{\mu-J,I}(y;x;-\kappa,-\delta)C_{\nu-I,J}(x+\delta I;y+\kappa(J-\mu);\delta,\kappa)}{C_{\mu-N,\emptyset}(y;x;-\kappa,-\delta)C_{\nu,N}(x;y+\kappa(N-\mu);\delta,\kappa)}.\end{align}
  Then, Theorem \ref{wr} is equivalent to the identity
  \begin{equation}\label{dwr}D_{\mu,\nu}(x;y)=0,\qquad |\mu|+|\nu|>0. \end{equation}
  
  We now insert \eqref{cm} into \eqref{dmn}. To distinguish  the shifted factorials  \eqref{sf}  with base $\delta$ from shifted factorials with base 
 $-\kappa$,  we use the notation
  $$[x;-\kappa]_k=[x][x-\kappa]\dotsm[x-(k-1)\kappa].$$
  By a straight-forward computation, 
  the factors involving only $y$-variables can be simplified to 
  \begin{multline*}\prod_{1\leq i<j\leq r}\frac{[y_i-y_j+(\mu_j-J_j-\mu_i+J_i)\kappa]}{[y_i-y_j+(\mu_j-N_j-\mu_i+N_i)\kappa]} \\
\begin{split}&\quad\times  \prod_{i,j=1}^r\frac{[y_i-y_j-\delta;-\kappa]_{\mu_i-J_i}[y_i-y_j-\kappa;-\kappa]_{\mu_i-N_i}}{[y_i-y_j-\kappa;-\kappa]_{\mu_i-J_i}[y_i-y_j-\delta;-\kappa]_{\mu_i-N_i}}
  \\
&\quad  \times \prod_{i\in J,\,j\in \langle r\rangle\setminus J}\frac{[y_i-y_j+(\mu_j-\mu_i+1)\kappa-\delta]}{[y_i-y_j+(\mu_j-\mu_i+1)\kappa]}\\
 &\quad  \times   \prod_{i\in N,\,j\in \langle r\rangle \setminus N}\frac{[y_i-y_j+(\mu_j-\mu_i+1)\kappa]}
   {[y_i-y_j+(\mu_j-\mu_i+1)\kappa-\delta]}
  \\
&=\prod_{i\in N\setminus J}\left( \prod_{j\in J}\frac{[y_i-y_j+(\mu_j-\mu_i-1)\kappa+\delta]}{[y_i-y_j+(\mu_j-\mu_i)\kappa]}
\prod_{j\in N}\frac{[y_i-y_j-(\mu_i-1)\kappa-\delta]}{[y_i-y_j-\mu_i\kappa]}\right).
\end{split}\end{multline*}
The factors involving both $x$- and $y$-variables are
\begin{multline*}
\prod_{i=1}^r\left(\prod_{j\in I}\frac{[y_i-x_j+\delta]}{[y_i-x_j-(\mu_i-J_i)\kappa]}\prod_{j\in\langle m\rangle\setminus I}\frac{[y_i-x_j+\kappa]}{[y_i-x_j-(\mu_i-J_i-1)\kappa]}
\right.\\
\begin{split}&\quad\times\left.\prod_{j=1}^m\frac{[y_i-x_j-(\mu_i-N_i-1)\kappa]}{[y_i-x_j+\kappa]}\right)\\
&\quad\times \prod_{i=1}^m\left(\prod_{j\in J}\frac{[x_i-y_j+ I_i\delta+(\mu_j-2)\kappa]}{[x_i-y_j+\nu_i\delta+(\mu_j-1)\kappa]}\prod_{j\in\langle r\rangle\setminus J}\frac{[x_i-y_j+(I_i-1)\delta+\mu_j \kappa]}{[x_i-y_j+(\nu_i-1)\delta+\mu_j \kappa]}\right.\\
&\quad\left.\times\prod_{j\in N}\frac{[x_i-y_j+\nu_i\delta+(\mu_j-1)\kappa]}{[x_i-y_j+(\mu_j-2)\kappa]}
\prod_{j\in \langle r\rangle\setminus N}\frac{[x_i-y_j+(\nu_i-1)\delta]}{[x_i-y_j-\delta]}\right)\\
&=\prod_{i\in I}\left(\prod_{j\in J}\frac{[x_i-y_j+\delta+(\mu_j-2)\kappa]}{[x_i-y_j+(\mu_j-1)\kappa]}\prod_{j\in N}\frac{[x_i-y_j-\delta]}{[x_i-y_j-\kappa]}\right)\\
&\quad\times 
\prod_{i\in N\setminus J}\left(
\prod_{j\in M\setminus I }\frac{[y_i-x_j+\delta-\mu_i\kappa]}{[y_i-x_j-(\mu_i-1)\kappa]}
\prod_{j\in M}\frac{[y_i-x_j-\nu_j\delta-(\mu_i-1)\kappa]}{[y_i-x_j-(\nu_j-1)\delta-\mu_i\kappa]}
\right)
.
\end{split}\end{multline*}
Finally, the factors involving only $x$-variables are
\begin{multline*}
\prod_{i\in I,\,j\in\langle m\rangle\setminus I}\frac{[x_i-x_j+\kappa]}{[x_i-x_j]}
\prod_{1\leq i<j\leq m}\frac{[x_i-x_j]}{[x_i-x_j+(I_i-I_j)\delta]}\\
\times\prod_{i,j=1}^m\frac{[x_i-x_j+\delta;\delta]_{\nu_i}[x_i-x_j+(I_i-I_j)\delta+\kappa;\delta]_{\nu_i-I_i}}{[x_i-x_j+\kappa;\delta]_{\nu_i}[x_i-x_j+(I_i-I_j+1)\delta;\delta]_{\nu_i-I_i}}\\
=\prod_{i\in I}\left(\prod_{j\in M\setminus I}\frac{[x_i-x_j+\delta-\kappa]}{[x_i-x_j]}
\prod_{j\in M}\frac{[x_i-x_j-\nu_j\delta]}{[x_i-x_j-(\nu_j-1)\delta-\kappa]}
\right)
.
\end{multline*}
We conclude that
\begin{multline}\label{dmne}D_{\mu,\nu}(x;y)= \sum_{I\subseteq M,\,J\subseteq N}(-1)^{[I|+|J|+|N|} \frac{[(\mu+|I|-|J|)\kappa+(|\nu|+|J|-|I|)\delta]}{[(|\mu|-|N|)\kappa+(|\nu|+|N|)\delta]}\\
\begin{split}&\times
\prod_{i\in I}\left(
\prod_{j\in M \setminus I}\frac{[x_i-x_j+\delta-\kappa]}{[x_i-x_j]}
\prod_{j\in J}\frac{[x_i-y_j+\delta+(\mu_j-2)\kappa]}{[x_i-y_j+(\mu_j-1)\kappa]}
\right.\\
&\times\left.\prod_{j\in M}\frac{[x_i-x_j-\nu_j\delta]}{[x_i-x_j-(\nu_j-1)\delta-\kappa]}\prod_{j\in N}\frac{[x_i-y_j-\delta]}{[x_i-y_j-\kappa]}\right)\\
&\times\prod_{i\in N\setminus J}\left(
\prod_{j\in M\setminus I}\frac{[y_i-x_j+\delta-\mu_i\kappa]}{[y_i-x_j-(\mu_i-1)\kappa]}
\prod_{j\in J}\frac{[y_i-y_j+(\mu_j-\mu_i-1)\kappa+\delta]}{[y_i-y_j+(\mu_j-\mu_i)\kappa]}
\right.\\
&\times
\left.
\prod_{j\in M}\frac{[y_i-x_j-\nu_j\delta-(\mu_i-1)\kappa]}{[y_i-x_j-(\nu_j-1)\delta-\mu_i\kappa]}
\prod_{j\in N}
\frac{[y_i-y_j-(\mu_i-1)\kappa-\delta]}{[y_i-y_j-\mu_i\kappa]}
\right).
\end{split}\end{multline}

We now explain how to identify \eqref{dwr} with a special case of
 the source identity \eqref{nssi}. As a first step, we write the index set 
 in \eqref{nssi} as a disjoint union
  $\langle n\rangle=M\sqcup N$. We make a corresponding change of variables
 $z_i\mapsto x_i$ for $i\in M$, $z_i\mapsto y_i$ for $i\in N$, $w_i\mapsto u_i$ for 
 $i\in M$ and $w_i\mapsto v_i$ for $i\in N$. 
  Finally, we make the substitutions $I\cap M\mapsto I$, $I^c\cap N\mapsto J$. 
  The left-hand side of \eqref{nssi} then takes the form
 \begin{multline*}
 \sum_{I\subseteq M,\,J\subseteq N}(-1)^{|I|+|J|+|N|}\frac{[|x|+|y|-|u|-|v|+(|I|+|N|-|J|)a]}{[|x|+|y|-|u|-|v|]}\\
\times \prod_{i\in I}\left(
\prod_{j\in M\setminus I}\frac{[x_i-x_j+a]}{[x_i-x_j]}\prod_{j\in J}\frac{[x_i-y_j+a]}{[x_i-y_j]}
\prod_{j\in M}\frac{[x_i-u_j]}{[x_i-u_j+a]}
\prod_{j\in N} \frac{[x_i-v_j]}{[x_i-v_j+a]}
 \right) \\
\times 
\prod_{i\in N\setminus J}\left(\prod_{j\in M\setminus I}\frac{[y_i-x_j+a]}{[y_i-x_j]}
 \prod_{j\in J}\frac{[y_i-y_j+a]}{[y_i-y_j]}\right.\\
 \left.\times\prod_{j\in M}\frac{[y_i-u_j]}{[y_i-u_j+a]}
 \prod_{j\in N}\frac{[y_i-v_j]}{[y_i-v_j+a]}
 \right).
 \end{multline*}
 Substituting
 $a\mapsto\delta-\kappa$ and, for all $i$, $x_i\mapsto x_i$, $y_i\mapsto y_i-(\mu_i-1)\kappa$, $u_i\mapsto x_i+\nu_i\delta$, $v_i\mapsto y_i+\delta$ in this expression gives
 \eqref{dmne}. This completes the proof of Theorem~\ref{wt}.

\section{Kernel function identities}
To prove Theorem \ref{kt} we will be need the following elliptic hypergeometric transformation formula. 

\begin{proposition}\label{ktp}
Assume that the parameters $x_1,\dots,x_m$, $y_1,\dots,y_r$, $X_1,\dots,X_n$, $Y_1,\dots,Y_s$ and
$a_1,\dots,a_{m+n}$  satisfy
\begin{equation}\label{bc}|x|+|a|+s\delta=|X|+r\delta. \end{equation}
Then,
\begin{multline}\label{kti}
\sum_{\substack{\mu\in\mathbb Z_{\geq 0}^m,\,I\subseteq \langle r\rangle,\\|\mu|+|I|=k}}
(-1)^{|I|}\prod_{1\leq i<j\leq m}\frac{[x_i-x_j+(\mu_i-\mu_j)\delta]}{[x_i-x_j]}
\\
\begin{split}&\quad\times\prod_{i\in I,\,j\in I^c}\frac{[y_i-y_j-\delta]}{[y_i-y_j]}\prod_{1\leq i\leq m,\,j\in I^c}
 \frac{[x_i-y_j-\delta]}{[x_i-y_j+(\mu_i-1)\delta]}
\\
&\quad\times\prod_{i=1}^m\left(\frac{\prod_{j=1}^{m+n}[x_i+a_j]_{\mu_i}}{\prod_{j=1}^m[x_i-x_j+\delta]_{\mu_i}\prod_{j=1}^n[x_i+X_j]_{\mu_i}}
\prod_{j=1}^s
\frac{[x_i+Y_j+\mu_i\delta]}{[x_i+Y_j]}\right)\\
&\quad\times\prod_{i\in I}\left(\frac{\prod_{j=1}^{m+n}[y_i+a_j]}{\prod_{j=1}^m[y_i-x_j-\mu_j\delta]\prod_{j=1}^n[y_i+X_j]}
\prod_{j=1}^s\frac{[y_i+Y_j+\delta]}{[y_i+Y_j]}\right)\\
&=
\sum_{\substack{\mu\in\mathbb Z_{\geq 0}^n,\,I\subseteq \langle s \rangle,\\|\mu|+|I|=k}}
(-1)^{|I|}\prod_{1\leq i<j\leq n}\frac{[X_i-X_j+(\mu_i-\mu_j)\delta]}{[x_i-x_j]}
\\
&\quad\times\prod_{i\in I,\,j\in I^c}\frac{[Y_i-Y_j-\delta]}{[Y_i-Y_j]}\prod_{1\leq i\leq n,\,j\in I^c}
 \frac{[X_i-Y_j-\delta]}{[X_i-Y_j+(\mu_i-1)\delta]}
\\
&\quad\times\prod_{i=1}^n\left(\frac{\prod_{j=1}^{m+n}[X_i-a_j]_{\mu_i}}{\prod_{j=1}^n[X_i-X_j+\delta]_{\mu_i}\prod_{j=1}^m[X_i+x_j]_{\mu_i}}
\prod_{j=1}^r
\frac{[X_i+y_j+\mu_i\delta]}{[X_i+y_j]}\right)\\
&\quad\times\prod_{i\in I}\left(\frac{\prod_{j=1}^{m+n}[Y_i-a_j]}{\prod_{j=1}^n[Y_i-X_j-\mu_j\delta]\prod_{j=1}^m[Y_i+x_j]}
\prod_{j=1}^r\frac{[Y_i+y_j+\delta]}{[Y_i+y_j]}\right).
\end{split}\end{multline}
\end{proposition}

Proposition \ref{ktp} is a  slight variation of a transformation formula found by
 Kajihara \cite{k} in the trigonometric case and, in general, in \cite{kn} and \cite{r2}. 
To be precise, that transformation appears as the special case $r=s=0$.
On the other hand, given that special case, the general case follows by substituting
$x\mapsto(x,y)$, $X\mapsto (X,Y)$ and
$$(a_1,\dots,a_{m+n})\mapsto(a_1,\dots,a_{m+n},y_1-\delta,\dots,y_r-\delta,Y_1+\delta,\dots,Y_s+\delta).$$

Alternatively, one can follow the approach of \cite{kn} and derive Proposition \ref{ktp}
from the  source identity
\eqref{ksni}. We find it instructive to sketch this proof. 
We start from \eqref{ksni}, with $a=\delta$ and $n$ replaced by $N$.
We first specialize the $z$-variables  as in \eqref{zs} and make a similar specialization 
\begin{multline*}
(w_1,\dots,w_N)\\
= \big(X_1,X_1+\delta,\dots,X_1+(\nu_1-1)\delta,\dots,
X_n,X_n+\delta,\dots,X_n+(\nu_n-1)\delta,
Y_1,\dots,Y_s\big).
\end{multline*}
 Here, we must have
\begin{equation}\label{dbc}N=|\lambda|+r=|\nu|+s. \end{equation}

Just as in the proof of Proposition \ref{csp}, the left-hand side of
\eqref{ksni} reduces to a sum over  $(\mu_1,\dots,\mu_m,P)$, 
where  $0\leq \mu_j\leq \lambda_j$ for each $j$, $P\subseteq \langle r\rangle$ and $|\mu|+|P|=k$. The resulting expression contains the product
$$F(\delta)=\prod_{i\in I,\,j\in I^c}\frac{[z_i-z_j-\delta]}{[z_i-z_j]},$$
which is computed in \eqref{fi}.
The remaining factors are easily computed in a similar way.
Apart from a sign factor $(-1)^k$,  the left-hand side of \eqref{zs} takes the form
\begin{multline*}
\sum_{\substack{\mu\in\mathbb Z_{\geq 0}^m,\,P\subseteq \langle r\rangle,\\
0\leq \mu_j\leq\lambda_j,\,1\leq j\leq m,\\
|\mu|+|P|=k}}
(-1)^{|P|}\prod_{1\leq i<j\leq m}\frac{[x_i-x_j+(\mu_i-\mu_j)\delta]}{[x_i-x_j]}\prod_{i,j=1}^m\frac{[x_i-x_j-\lambda_j\delta]_{\mu_i}}{[x_i-x_j+\delta]_{\mu_i}}\\
\times\prod_{i=1}^m\left(
  \prod_{j\in P}
 \frac{[x_i-y_j+\lambda_i\delta]}{[x_i-y_j+\mu_i\delta]}
 \prod_{j\in P^c}
 \frac{[x_i-y_j-\delta]}{[x_i-y_j+(\mu_i-1)\delta]}
 \right)\prod_{i\in P,\,j\in P^c}\frac{[y_i-y_j-\delta]}{[y_i-y_j]}\\
\times\prod_{\substack{1\leq i\leq m,\\ 1\leq j\leq n}}\frac{[x_i+X_j+\nu_j\delta]_{\mu_i}}{[x_i+X_j]_{\mu_i}}
\prod_{\substack{1\leq i\leq m,\\1\leq j\leq s}}
\frac{[x_i+Y_j+\mu_i\delta]}{[x_i+Y_j]}\\
\times\prod_{i\in P,\,1\leq j\leq n}\frac{[y_i+X_j+\nu_j\delta]}{[y_i+X_j]}\prod_{i\in P,\,1\leq j\leq s}\frac{[y_i+Y_j+\delta]}{[y_i+Y_j]}.
\end{multline*}
Here, the restrictions  $\mu_j\leq\lambda_j$ may be ignored, since 
$[x_i-x_j-\lambda_j\delta]_{\mu_i}$ vanishes if $i=j$ and $\mu_j>\lambda_j$. 
We then obtain the left-hand side of \eqref{kti}, in the special case when
\begin{equation}\label{as}(a_1,\dots,a_{m+n})=(-x_1-\lambda_1\delta,\dots,-x_m-\lambda_m\delta,
X_1+\nu_1\delta,\dots,X_n+\nu_n\delta). \end{equation}
By \eqref{dbc}, this is consistent with the balancing condition \eqref{bc}.
It is clear from symmetry considerations that the right-hand side of 
\eqref{ksni} reduces to the corresponding right-hand side of \eqref{kti}.
We conclude that \eqref{kti} holds in the infinitely many special cases \eqref{as},
with $\lambda_j$ and $\nu_j$ non-negative integers subject to \eqref{dbc}. 
Finally, by the same type of analytic continuation argument that was used in the proof
of Proposition \ref{lsl}, \eqref{kti} holds for general values of  $a_j$, as long as \eqref{bc} is satisfied.
This proves Proposition \ref{ktp}.

%Finally, one can apply an argument of analytic continuation to show that \eqref{kti} holds for general values of the parameters $a_j$. More precisely, 
%suppose that we eliminate $a_{m+n}$ using \eqref{bc} and consider the difference $f$
%of the left-hand and  right-hand side of \eqref{kti} as a function of $a_1$. Each term in $f$ has the form
%$$C\prod_{l=1}^{k}[a_1-t_l][a_1+t_l-\lambda], $$
%where $\lambda=a_1+a_{m+n}$ is considered as a constant. 
%In the generic situation \eqref{ws}, it follows that  $f$ is a quasi-periodic function with respect to the underlying lattice $\Gamma=\mathbb Z\omega_1+\mathbb Z\omega_2$. 
%This is also true for the variables $a_2,\dots,a_{m+n-1}$. By our assumption \eqref{gp}, all the points \eqref{as}
%are distinct modulo $\Gamma$. It follows that $f$ has infinitely many zeroes in each quasi-period parallelogram and hence vanishes identically.
%This proves Proposition \ref{ktp}.

We now turn to the proof of Theorem \ref{kt}. We write the kernel function identity as
\begin{multline} \label{kfm}\sum_{\substack{\mu\in\mathbb Z_{\geq 0}^m,\,I\subseteq \langle r\rangle,\\|\mu|+|I|=k}} 
C_{\mu,I}(x;y)\frac{\Phi(x+\delta\mu;y-\kappa I;X;Y)}{\Phi(x;y;X;Y)}\\
=
\sum_{\substack{\mu\in\mathbb Z_{\geq 0}^n,\,I\subseteq \langle s \rangle,\\|\mu|+|I|=k}} 
C_{\mu,I}(X;Y)\frac{\Phi(x;y;X+\delta\mu;Y-\kappa I)}{\Phi(x;y;X;Y)}.\end{multline}
It is straight-forward to compute
\begin{multline*}\frac{\Phi(x+\delta\mu;y-\kappa I;X;Y)}{\Phi(x;y;X;Y)}
=\prod_{i=1}^m\left(\prod_{j=1}^n\frac{[x_i+X_j-\kappa]_{\mu_i}}{[x_i+X_j]_{\mu_i}} \prod_{j=1}^s
\frac{[x_i+Y_j+\delta\mu_i]}{[x_i+Y_j]}
\right)\\
\times\prod_{{i\in I}}\left(\prod_{j=1}^n\frac{[y_i+X_j-\kappa]}{[y_i+X_j]}\prod_{j=1}^s\frac{[y_i+Y_j+\delta]}{[y_i+Y_j]}
\right).
\end{multline*}
Inserting \eqref{cm}, the left-hand side of \eqref{kfm} is 
\begin{multline*}\sum_{\substack{\mu\in\mathbb Z_{\geq 0}^m,\,I\subseteq \langle r\rangle,\\|\mu|+|I|=k}} 
(-1)^{|I|}\prod_{1\leq i<j\leq m}\frac{[x_i-x_j+(\mu_i-\mu_j)\delta]}{[x_i-x_j]}
\prod_{i\in I,\,j\in I^c}\frac{[y_i-y_j-\delta]}{[y_i-y_j]}\\
\times\prod_{i,j=1}^m\frac{[x_i-x_j+\kappa]_{\mu_i}}{[x_i-x_j+\delta]_{\mu_i}}
\prod_{i=1}^m\left(\prod_{j\in I}\frac{[x_i-y_j-\kappa]}{[x_i-y_j+\mu_i\delta]}\prod_{j\in I^c}\frac{[x_i-y_j-\delta]}{[x_i-y_j+(\mu_i-1)\delta]}\right)\\
\times\prod_{i=1}^m\left(\prod_{j=1}^n\frac{[x_i+X_j-\kappa]_{\mu_i}}{[x_i+X_j]_{\mu_i}} \prod_{j=1}^s
\frac{[x_i+Y_j+\delta\mu_i]}{[x_i+Y_j]}
\right)\\
\times\prod_{{i\in I}}\left(\prod_{j=1}^n\frac{[y_i+X_j-\kappa]}{[y_i+X_j]}\prod_{j=1}^s\frac{[y_i+Y_j+\delta]}{[y_i+Y_j]}
\right).
\end{multline*}
This agrees with the left-hand side of \eqref{kti}, under the specialization
$$(a_1,\dots,a_{m+n})=(\kappa-x_1,\dots,\kappa-x_m,X_1-\kappa,\dots,X_m-\kappa). $$
Note that the balancing condition \eqref{bc} reduces to \eqref{kfc} in this case.
By symmetry, the right-hand side of \eqref{kfm} reduces to the corresponding right-hand side of \eqref{kti}. This proves Theorem \ref{kt}.

\appendix 
\section{Relation to deformed Ruijsenaars model} 
\label{appA}

The conventions used in this paper differ from the ones that are more common in the physics literature, going back to the work of Ruijsenaars \cite{r}.
For the convenience of the reader, we explain the relation between these conventions. In particular, we give the precise relation  between the operators $H_{m,r}^{(1)}$ and the deformed Ruijsenaars model introduced in \cite{ahl}.

The Ruijsenaars systems are defined by two difference operators, $\hS^+$ and $\hS^-$, defining a Hamiltonian $\hH=\hS^++\hS^-$ and a momentum operator $\hP=\hS^+-\hS^-$ which, together with a boost operator $\hB$, 
represent the Poincar\'e algebra in 1+1 spacetime dimensions. That is, the commutation relations 
\begin{equation}
\label{hpb} 
[\hH,\hP]=0,\quad [\hH,\hB]=\ti \hP,\quad [\hP,\hB]=\ti \hH .
\end{equation} 
are satisfied \cite{r}. 
In particular, for the deformed elliptic Ruijsenaars model, the corresponding operators  are given by (we rename $(\beta,\beta g)$ in \cite[Eq.\ (16)]{ahl} to $(\ti\delta,\ti\kappa)$ and drop an irrelevant overall constant)
\begin{equation*} 
\label{hS}
\hS^{\pm} = \sum_{i=1}^m \frac{[\kappa]}{[\delta]}A_i^{\mp}e^{\pm\delta\frac{\partial}{\partial x_i}}A_i^{\pm} - \sum_{i=1}^r B_i^{\mp}e^{\mp\kappa\frac{\partial}{\partial y_i}}B_i^{\pm},
\end{equation*} 
where
\begin{align*} 
A_i^\pm& = \prod_{\substack{1\leq j\leq m\\j\neq i}} \left(\frac{[x_i-x_j\mp\kappa ]}{[x_i-x_j]}\right)^{1/2}  \prod_{j=1}^r \left( \frac{[x_i-y_j\mp\kappa/2\pm\delta/2]}{[x_i-y_j\mp\kappa /2\mp\delta /2]}\right)^{1/2},\\
B_i^\pm& =\prod_{\substack{1\leq j\leq r\\j\neq i}} \left(\frac{[y_i-y_j\pm\delta ]}{[y_i-y_j]}\right)^{1/2}  \prod_{j=1}^m \left( \frac{[y_i-x_j\pm\delta/2\mp\kappa /2]}{[y_i-x_j\pm\delta/2\pm\kappa/2]}\right)^{1/2}. 
\end{align*} 
One can check that $\hH=\hS^++\hS^-$, $\hP=\hS^+-\hS^-$, together with 
\begin{equation*} 
\hB = \frac{\ti}{\delta} \sum_{i=1}^m x_i -\frac{\ti}{\kappa} \sum_{i=1}^ry_i , 
\end{equation*} 
indeed satisfy \eqref{hpb}. 

We will now show that,
up to a similarity transformation and shifts of the variables, the operators $\hS^+$ and $\hS^-$
 are equal to, respectively, our operators $H^{(1)}_{m,r}$ and $\hat H^{(1)}_{m,r}$.
%The parameter $\beta>0$ and $g>0$ in the Ruijsenaars model, which have a clear physical interpretation \cite{r},  are related to our as follows, 
%\begin{equation*} 
%\delta=-\ti\beta,\quad \kappa = -\ti\beta g.
%\end{equation*} 
%Since our results hold true for all complex parameters satisfying the conditions in (2.3), they apply to the cases of  interest in physics. 
To this end, we introduce the function
\begin{align*} 
\Delta &=  \prod_{\substack{1\leq i,j\leq m\\i\neq j}}\frac{G_\delta(x_i-x_j+\kappa)}{G_\delta(x_i-x_j)}  \prod_{\substack{1\leq i,j\leq r\\i\neq j}}\frac{G_{-\kappa}(y_i-y_j-\delta)}{G_{-\kappa}(y_i-y_j)}
\\ 
&\quad\times \prod_{i=1}^m\prod_{j=1}^r \frac1{[x_i-y_j +\kappa/2-\delta/2][y_i-x_i +\kappa/2-\delta/2]}.
\end{align*} 
A straight-forward computation gives
\begin{equation*} 
\begin{split} 
\Delta^{-1/2}\hS^{\pm}\Delta^{1/2} = & \frac{[\kappa]}{[\delta]} \sum_{i=1}^m  \prod_{\substack{1\leq j\leq m\\j\neq i}} \frac{[x_i-x_j\pm\kappa ]}{[x_i-x_j]} 
 \prod_{j=1}^r\frac{[x_i-y_j\pm\kappa/2\mp\delta/2]}{[x_i-y_j\pm\kappa /2\pm \delta /2]} 
e^{\pm\delta\frac{\partial}{\partial x_i}}\\ 
&- \sum_{i=1}^r \prod_{\substack{1\leq j\leq r\\j\neq i}}\frac{[y_i-y_j\mp\delta ]}{[y_i-y_j]} 
  \prod_{j=1}^m \frac{[y_i-x_j\mp\delta/2\pm\kappa /2]}{[y_i-x_j\mp\delta/2\mp\kappa/2]} e^{\mp\kappa\frac{\partial}{\partial y_i}}.
\end{split} 
\end{equation*} 
Moreover, the case $k=1$ of \eqref{docm} and \eqref{hh} can be written
\begin{equation*} 
\begin{split} 
H^{(1)}_{m,r} = &\frac{[\kappa]}{[\delta]}\sum_{i=1}^m \prod_{\substack{1\leq j\leq m\\j\neq i}} \frac{[x_i-x_j+\kappa]}{[x_i-x_j]} \prod_{j=1}^r \frac{[x_i-y_j-\delta]}{[x_i-y_j]}e^{\delta\frac{\partial}{\partial x_i}}\\
& -  \sum_{i=1}^r \prod_{\substack{1\leq j\leq r\\j\neq i}}\frac{[y_i-y_j-\delta]}{[y_i-y_j]} \prod_{j=1}^m \frac{[y_i-x_j+\kappa]}{[y_i-x_j]}e^{-\kappa\frac{\partial}{\partial y_i}},
\end{split}
\end{equation*} 
\begin{equation*} 
\begin{split} 
\hat H^{(1)}_{m,r} = & \frac{[\kappa]}{[\delta]}\sum_{i=1}^m \prod_{\substack{1\leq j\leq m\\j\neq i}}\frac{[x_i-x_j-\kappa]}{[x_i-x_j]} \prod_{j=1}^r \frac{[x_i-y_j-\kappa]}{[x_i-y_j-\delta-\kappa]}e^{-\delta\frac{\partial}{\partial x_i}}\\
& -  \sum_{i=1}^r \prod_{\substack{1\leq j\leq r\\j\neq i}}\frac{[y_i-y_j+\delta]}{[y_i-y_j]} \prod_{j=1}^m \frac{[y_i-x_j+\delta]}{[y_i-x_j+\kappa+\delta]}e^{\kappa\frac{\partial}{\partial y_i}}.
\end{split}
\end{equation*} 
This makes manifest that, after shifting the variables in $\Delta^{-1/2}\hS^{\pm}\Delta^{1/2}$ as 
\begin{equation} \label{vs}
x_i\to x_i \pm  \delta/2,\quad y_j\to y_j\mp \kappa/2\quad (i=1,\ldots,m, j=1,\ldots,r), 
\end{equation} 
one obtains the operators $H^{(1)}_{m,r}$ and $\hat H^{(1)}_{m,r}$, respectively. 

It is interesting to note  that $\Delta$ is the weight function in a natural scalar product on the space of common eigenfunctions of the operators $\hS^\pm$ proposed recently in \cite{AHL20}.

Finally, we comment on the role of the condition \eqref{gp}. In the elliptic case, we can normalize the function $[x]$ so that its zero set is $\mathbb Z+\tau\mathbb Z$ for some $\tau$ in the upper half-plane.  Then, \eqref{gp} means that
$$\delta,\ \kappa\notin \mathbb Q+\tau\mathbb Q. $$
The physically most natural case is when $\tau\in\ti\mathbb R_{>0}$ and the parameters $\beta=\ti\delta$ and $g=\kappa/\delta$ are real. This gives the conditions
$$\beta,\ g\beta\notin \ti\tau\mathbb Q. $$
We need these conditions to make the operators $H_{m,r}^{(k)}$ and $D_{m,r}^{(k)}$ well-defined for all $k$. If they are violated, the operators still make sense for a finite range of $k$, which could conceivably be extended by appropriate renormalization.

\section{Multiplicative notation}
\label{multapp}

We have considered our  operators as acting by additive shifts. In the trigonometric and elliptic cases, they can also be realized by multiplicative shifts. We will restate our main results in this form, as it is very common in the literature.

Excluding the rational case, the function  $x\mapsto [x]$ can be chosen as periodic. After rescaling the variable, we can assume that the primitive period is $2$. We then normalize the function as
\begin{equation}\label{fx}[x]=e^{-\ti\pi x}\theta(e^{2\ti\pi x};p), \end{equation}
where
$$\theta(z;p)=\prod_{j=0}^\infty(1-p^j z)\left(1-\frac{p^{j+1}}z\right) $$ 
and the elliptic nome $p$  satisfies $|p|<1$.
The  trigonometric case is included as
$$ [x]_{p=0}=e^{-\ti\pi x}(1-e^{2\ti\pi x})=-2\ti\sin(\pi x).$$
If $z=e^{2\ti\pi x}$, the additive shifts $x\mapsto x+\delta$ and $x\mapsto x-\kappa$ correspond to  $z\mapsto qz$, $z\mapsto t^{-1} z$, where
$$q=e^{2\ti\pi \delta},\qquad t=e^{2\ti\pi\kappa}. $$
The assumption \eqref{gp} means that $q,\,t\notin p^{\mathbb Q}$.

Consider the operators $H_{m,r}^{(k)}$ as acting on functions that are $1$-periodic in the variables $x_j$ and $y_j$, and hence can be expressed in terms of  $z_j=e^{2\ti \pi x_j}$ and $w_j=e^{2\ti\pi y_j}$. 
We will normalize the resulting multiplicative difference operator as
\begin{subequations}\label{ao}
\begin{equation}\label{hm}\mathbf H_{m,r}^{(k)}=
\mathbf H_{m,r}^{(k)}(z_1,\dots,z_m;w_1,\dots,w_r;q,t)=
e^{\ti\pi k\left((r-1)\delta-m\kappa\right)}H_{m,r}^{(k)}. \end{equation}
We also introduce the modified operators 
\begin{align}\label{dmm}\mathbf D_{m,r}^{(k)}&=
\mathbf H_{r,m}^{(k)}(w_1,\dots,w_r;z_1,\dots,z_m;t^{-1},q^{-1}),\\
 \hat {\mathbf H}_{m,r}^{(k)}&=\mathbf H_{m,r}^{(k)}(q^{-1}z_1,\dots,q^{-1}z_m;tw_1,\dots,tw_r;q^{-1},t^{-1}),\\
\hat{\mathbf   D}_{m,r}^{(k)}&=\mathbf H_{r,m}^{(k)}(tw_1,\dots,tw_r;q^{-1}z_1,\dots,q^{-1}z_m;t,q),
\end{align}
\end{subequations}
which are related to the additive operators used in the main text by
\begin{align*}\mathbf{D}_{m,r}^{(k)}&=e^{\ti\pi k\left(r\delta-(m-1)\kappa\right)} D_{m,r}^{(k)},\\
\hat{\mathbf H}_{m,r}^{(k)}&=e^{\ti\pi k\left(m\kappa-(r-1)\delta\right)} \hat H_{m,r}^{(k)},\\
\hat{\mathbf D}_{m,r}^{(k)}&=e^{\ti\pi k\left((m-1)\kappa-r\delta\right)} \hat D_{m,r}^{(k)}.
\end{align*}

It is straight-forward to verify that, in the notation
$$(a;q,p)_k=\theta(a;p)\theta(a q;p)\dotsm \theta(aq^{k-1};p), $$
$$T_{q,z}^\mu f(z_1,\dots,z_m)=f(q^{\mu_1}z_1,\dots,q^{\mu_m}z_m),$$
we have
$$\mathbf H_{m,r}^{(k)}=\sum_{\substack{\mu\in\mathbb Z_{\geq 0}^m,\,I\subseteq \langle r\rangle,\\|\mu|+|I|=k}} 
 C_{\mu,I}(z;w)\,T_{q,z}^{\mu}T_{t^{-1},w}^{ I},$$
 where
 \begin{multline*} C_{\mu,I}(z;w)=(-1)^{|I|}(t^{-m}q^r)^{|\mu|}q^{\binom{|I|}2}
 \prod_{i\in I,\,j\in I^c}\frac{\theta(qw_j/w_i;p)}{\theta(w_j/w_i;p)}\prod_{i,j=1}^m\frac{(tz_i/z_j;q,p)_{\mu_i}}{(qz_i/z_j;q,p)_{\mu_i}}
 \\
\times
\prod_{1\leq i<j\leq m}\frac{q^{\mu_j}\theta(q^{\mu_i-\mu_j}z_i/z_j;p)}{\theta(z_i/z_j;p)}
\prod_{i=1}^m\left(\prod_{j\in I}\frac{\theta(z_i/tw_j;p)}{\theta(q^{\mu_i}z_i/w_j;p)}\prod_{j\in I^c}\frac{\theta(z_i/qw_j;p)}{\theta(q^{\mu_i-1}z_i/w_j;p)}\right).
\end{multline*}

In multiplicative notation, Theorem \ref{ait}, Theorem \ref{wt}, Corollary \ref{wc}, Corollary~\ref{dc} and Corollary \ref{fcc} can be summarized as follows.

\begin{theorem}
For fixed $m$ and $r$, the four infinite families of operators \eqref{ao} mutually commute. 
If $q$ and $t$ are generic, the operators \eqref{hm} are algebraically independent for $1\leq k\leq m+r$.
The operators \eqref{hm} and \eqref{dmm} are related by
$$\sum_{k+l=N} t^k\theta(q^k t^l)\mathbf H_{m,r}^{(k)}\mathbf D_{m,r}^{(l)}=0, \qquad N\geq 1,$$
and by
$$\mathbf {H}_{m,r}^{(l)}=(-t)^{-l}\det_{1\leq i,j\leq l}\left(\frac{\theta(t^{i-j+1}q^{j-1})}{\theta(q^i)}\,{\mathbf{ D}}_{m,r}^{(i-j+1)}\right), $$
where one should interpret matrix elements with $i-j+1<0$ as zero. 
\end{theorem}

To write Theorem \ref{kt} in multiplicative notation takes some more work.
We will express the kernel function in terms of the elliptic gamma function \cite{ru2}
$$\Gamma(z;p,q)=\prod_{j,k=0}^\infty\frac{1-p^{j+1}q^{k+1}/z}{1-p^jq^k z}, \qquad |p|,\,|q|<1,$$
which satisfies the $q$-difference equation
$$ \frac{\Gamma(qz;p,q)}{\Gamma(z;p,q)}=\theta(z;p).$$
Equivalently, the function
\begin{subequations}\label{gd}
\begin{equation}\label{gda}G_\delta(x)=e^{\frac{\ti\pi x(\delta-x)}{2\delta}}\Gamma(e^{2\ti\pi x};p,q) \end{equation}
satisfies \eqref{gfe}. This solution is valid for $|q|<1$, that is, $\operatorname{Im}(\delta)>0$.  If $\operatorname{Im}(\delta)<0$, one can instead take
\begin{equation}G_\delta(x)=\frac {e^{\frac{\ti\pi x(\delta-x)}{2\delta}}}{\Gamma(q^{-1}e^{2\ti\pi x};p,q^{-1})}. \end{equation}
\end{subequations}
In either case, the general solution of \eqref{gfe} is $G_\delta$ times an arbitrary $\delta$-periodic meromorphic function.
The construction of solutions to \eqref{gfe} with real $\delta$ (that is, $|q|=1$)
is more complicated \cite{s}, so  we will assume for simplicity that $|q|,\,|t|\neq 1$. 
In the case 
$|q|<1<|t|$, we introduce the multiplicative kernel function
\begin{multline}\label{mkf}
\mathbf \Phi^{(m,r,n,s)}(z_1,\dots,z_m;w_1,\dots,w_r;Z_1,\dots,Z_n;W_1,\dots,W_s)\\
=\prod_{\substack{1\leq i\leq m,\\1\leq j\leq n}}\frac{\Gamma(t^{-1}z_iZ_j;p,q)}{\Gamma(z_iZ_j;p,q)}
\prod_{\substack{1\leq i\leq r,\\1\leq j\leq s}}\frac{\Gamma(qw_iW_j;p,t^{-1})}{\Gamma(w_iW_j;p,t^{-1})}\\
\times\prod_{\substack{1\leq i\leq m,\\1\leq j\leq s}}\theta(z_iW_j;p)
\prod_{\substack{1\leq i\leq r,\\1\leq j\leq n}}\theta(w_i Z_j;p).
\end{multline}
If  one or both the parameters $|q|$ and $|t^{-1}|$  is larger than $1$,
we define $\mathbf \Phi$ by the expression obtained from \eqref{mkf} after making
  the formal replacement
$$\Gamma(x;p,s)\mapsto\frac 1{\Gamma(sx;p,1/s)},\qquad |s|>1. $$

\begin{theorem}\label{mkt}
Assuming $ t^{m-n}=q^{r-s}$, the kernel function identity
\begin{equation}\label{mki}\mathbf H_{m,r}^{(k)}
{(z;w)}\mathbf \Phi^{(m,r,n,s)}(z;w;Z;W)
=\mathbf H_{n,s}^{(k)}
{(Z;W)}\mathbf \Phi^{(m,r,n,s)}(z;w;Z;W)\end{equation}
holds.
\end{theorem}

 In particular, \eqref{mki} holds if
$m=n$ and $r=s$.

To prove Theorem \ref{mkt}, we insert  \eqref{fx} and \eqref{gd} into  \eqref{ph}.
In terms of the multiplicative variables $z_j=e^{2\ti\pi x_j}$, $w_j=e^{2\ti\pi y_j}$, $Z_j=e^{2\ti\pi X_j}$ and $W_j=e^{2\ti\pi Y_j}$,  the additive and multiplicative kernel functions are related by
\begin{multline*}\frac{\Phi^{(m,r,n,s)}(x;y;X;Y)}{\mathbf\Phi^{(m,r,n,s)}(z;w;Z;W)}\\
=\prod_{\substack{1\leq i\leq m,\\ 1\leq j\leq n}}\frac{e^{\ti\pi(x_i+X_j-\kappa)(\delta-x_i-X_j+\kappa)/2\delta}}{e^{\ti\pi(x_i+X_j)(\delta-x_i-X_j)/2\delta}} 
\prod_{\substack{1\leq i\leq r,\\ 1\leq j\leq s}}\frac{e^{-\ti\pi(y_i+Y_j+\delta)(-\kappa-y_i-Y_j-\delta)/2\kappa}}{e^{-\ti\pi(y_i+Y_j)(-\kappa-x_i-X_j)/2\kappa}} 
\\
\times
\prod_{\substack{1\leq i\leq m,\\ 1\leq j\leq s}}e^{-\ti\pi(x_i+Y_j)}
\prod_{\substack{1\leq i\leq r,\\ 1\leq j\leq n}}e^{-\ti\pi(y_i+X_j)}
=Ce^{\ti\pi(A(x;y)+B(X;Y))},\end{multline*}
where  $C$ is an irrelevant constant and
\begin{align*}
A(x;y)&=\frac{\kappa}{\delta}\,n|x|+\frac{\delta}{\kappa}\,s|y|-s|x|-n|y|,\\
B(X;Y)&=\frac{\kappa}{\delta}\,m|X|+\frac{\delta}{\kappa}\,r|Y|-r|X|-m|Y|.
\end{align*}

We  can now write the kernel function identity \eqref{kfi} as
\begin{multline}\label{ckf}e^{-\ti\pi A(x;y)}H_{m,r}^{(k)}{(x;y)}e^{\ti\pi A(x;y)}\mathbf \Phi^{m,r,n,s}(z;w;Z;W)\\
= e^{-\ti\pi B(X;Y)}H_{n,s}^{(k)}{(X;Y)}e^{\ti\pi B(X;Y)}\mathbf \Phi^{(m,r,n,s)}(z;w;Z;W).\end{multline}
The operator on the left is a sum of terms of the form
$$ e^{-\ti\pi A(x;y)}T_x^{\delta\mu}T_y^{-\kappa I} e^{\ti\pi A(x;y)}
=e^{\ti\pi\left(A(x+\delta\mu;y-\kappa I)-A(x;y)\right)}T_x^{\delta\mu}T_y^{-\kappa I} =e^{\ti\pi k(n\kappa -s\delta )}T_x^{\delta\mu}T_y^{-\kappa I} .
$$
Hence,
$$ e^{-\ti\pi A(x;y)} H_{m,r}^{(k)} e^{\ti\pi A(x;y)}=e^{\ti\pi k(n\kappa -s\delta )}H_{m,r}^{(k)}
=e^{\ti\pi k((m+n)\kappa - (r+s-1)\delta)}\mathbf H_{m,r}^{(k)}
.$$
On the right-hand side of \eqref{ckf}, the same exponential prefactor appears and can be canceled. This proves Theorem \ref{mkt}.

 \end{document}